\documentclass[letterpaper,UKenglish,numberwithinsect]{lipics-v2018}

\usepackage{microtype}
\usepackage{wrapfig}
\usepackage{Tabbing}

\theoremstyle{definition}
\newtheorem{algorithm}[theorem]{Algorithm}

\graphicspath{{./figures/}}
\usepackage[usenames,dvipsnames]{xcolor}
\definecolor{myblue}{RGB}{0,102,202}
\definecolor{myred}{RGB}{213,0,73}

\usepackage{tikz}
\usepackage{pgffor}
\pgfdeclarelayer{bg}
\pgfsetlayers{bg,main}
\newcommand{\curveAnnot}[3][black]{%
	\begin{tikzpicture}[remember picture,x=.2cm,y=.04cm,baseline={([yshift={-\ht\strutbox}]current bounding box.north)},outer sep=0pt,inner sep=0pt]
		\edef\yp{0}
		\edef\y{0}
		\edef\xp{0}
		\foreach \x [count=\xi] in {#3} {
			\ifthenelse{\NOT 1 = \xi}{
				\draw[#1,cap=round] (\xp,\yp) -- (\xp,\y);
				\draw[#1,cap=round] (\xp,\y) -- (\x,\y);
				\xdef\yp{\y}
				\pgfmathparse{\y+1}
				\xdef\y{\pgfmathresult}
			}{}
			\xdef\xpp{\xp}
			\xdef\xp{\x}
		}
		\pgfmathparse{\xp+.001*(\xpp-\xp)}
		\draw[#1,->,cap=round] (\pgfmathresult,\yp) -- (\xp,\yp);
		#2
		\begin{pgfonlayer}{bg}
			\foreach \x in {0,...,11} {\draw[gray,dotted] (\x,-1) -- (\x,\y);}
		\end{pgfonlayer}
	\end{tikzpicture}}
\newcommand{\curve}[2][black]{{\curveAnnot[#1]{\empty}{#2}}}

\newcommand{\myparNS}[1]{\noindent{\sffamily\bfseries #1.}}
\newcommand{\mypar}[1]{\medskip\myparNS{#1}}

\newcommand{\SETH}{\textsc{SETH}}
\newcommand{\YES}{\textsc{Yes}}
\newcommand{\NO}{\textsc{No}}
\newcommand{\uOV}{$\textsc{Orthog}^*$}
\newcommand{\OV}{\textsc{Orthog}}

\newcommand{\alg}{\textbf{Alg}}
\newcommand{\bigO}{\mathcal{O}}
\newcommand{\R}{\mathbb{R}}
\newcommand{\N}{\mathbb{N}}
\newcommand{\C}[1]{\langle #1\rangle}
\newcommand{\from}{\colon}
\newcommand{\e}{\varepsilon}
\newcommand{\width}{w}

\newcommand{\dF}{d_F}
\newcommand{\ddF}{d_{dF}}
\newcommand{\dwF}{d_{wF}}
\newcommand{\dwwF}{d_{wwF}}
\newcommand{\doF}{d_{\vec{F}}}
\newcommand{\rev}[1]{\ensuremath{\mathit{reverse}(#1)}}
\renewcommand{\u}{{\vec{u}}}
\renewcommand{\v}{{\vec{v}}}

\newcommand{\enumi}[1]{\textcolor{darkgray}{\sffamily\bfseries\upshape\mathversion{bold}#1.}}

\title{SETH Says: Weak Fr\'echet Distance is Faster, but only if it is Continuous and in One Dimension}
\titlerunning{Weak Fr\'echet Distance is Faster if it is Continuous and in One Dimension}

\author{Kevin Buchin, Tim Ophelders, Bettina Speckmann}
	{Dep. of Mathematics and Computer Science, TU Eindhoven, The Netherlands\\{}
	[\href{mailto:k.a.buchin@tue.nl}{k.a.buchin}|\href{mailto:t.a.e.ophelders@tue.nl}{t.a.e.ophelders}|\href{mailto:b.speckmann@tue.nl}{b.speckmann}]@tue.nl}
	{\empty}
	{\empty}
	{The authors are supported by the Netherlands Organisation for Scientific Research (NWO) under project no.~612.001.207 (Kevin Buchin) and no.~639.023.208 (Tim Ophelders and Bettina Speckmann).}

\authorrunning{K.\ Buchin and T.\ Ophelders and B.\ Speckmann}

\Copyright{Kevin Buchin and Tim Ophelders and Bettina Speckmann}


\ArticleNo{0}
\nolinenumbers 
\hideLIPIcs  

\begin{document}

\maketitle

\begin{abstract}
	We show by reduction from the Orthogonal Vectors problem that algorithms with strongly subquadratic running time cannot approximate the Fr\'echet distance between curves better than a factor~$3$ unless SETH fails.
	We show that similar reductions cannot achieve a lower bound with a factor better than~$3$.
	Our lower bound holds for the continuous, the discrete, and the weak discrete Fr\'echet distance even for curves in one dimension.
	Interestingly, the continuous weak Fr\'echet distance behaves differently.
	Our lower bound still holds for curves in two dimensions and higher.
	However, for curves in one dimension, we provide an exact algorithm to compute the weak Fr\'echet distance in linear time.
\end{abstract}

\section{Introduction}

	The \emph{Fr\'echet distance} is a popular metric for measuring the similarity between curves. Intuitively, it measures how well two parameterized curves can be aligned by a monotone reparameterization. The Fr\'echet distance finds many applications, in particular in the analysis and visualization of movement data~\cite{brakatsoulas2005map,BuchinBDFJSSSW17,GudmundssonWolle10,konzack2017visual}. 
	Alt and Godau~\cite{altgodau} were the first to study the Fr\'echet distance from a computational perspective. They presented an algorithm that computes the Fr\'echet distance between two polygonal curves of complexity $n$ in $\bigO(n^2 \log n)$ time. 
	Alt and Godau's work triggered a wealth of research on the Fr\'echet distance. Specific topics of interest include algorithms to compute the Fr\'echet distance for special classes of curves~\cite{AronovHKWW06,DriemelHW12}, generalizations to surfaces ~\cite{AltBuchin10,BuchinBuSc10,NayyeriX16}, and algorithms for meaningful variants~\cite{BuchinBW09,CookWenk10,DriemelH13}. 
	
	Despite all these results, the bound of $\bigO(n^2 \log n)$ by Alt and Godau for the original problem of computing the Fr\'echet distance between two general polygonal curves stood for nearly twenty years. Only quite recently there has finally been progress on this question. First, Buchin et al.~\cite{buchin2017soviets} presented an algorithm with a slightly improved (but still super-quadratic) running time. Then, Bringmann~\cite{Bringmann14} proved that no significantly faster algorithm for computing the Fr\'echet distance between two general polygonal curves exists unless the \emph{Strong Exponential Time Hypothesis} (\SETH{}) fails. Bringmann's proof nearly settles the question, except for one important special case: curves in one dimension. His construction uses curves embedded in two-dimensional space and hence the question remained open whether a similar conditional lower bound holds also in one dimension.
	
	One dimensional curves (parameterized over time) naturally occur in time series analysis. In this context the Fr\'echet distance can, for instance, be used to cluster data~\cite{DriemelKS16}. The Fr\'echet distance in one dimension can also be used as a subroutine for approximating the Fr\'echet distance for curves in two and higher dimensions~\cite{Bringmann17}. Bringmann's lower bound sparked renewed interest in the computation of the Fr\'echet distance between one-dimensional curves. Cabello and Korman showed that for two $1$D curves that do not overlap, the Fr\'echet distance can be computed in linear time (personal communication, referenced in~\cite{Bringmann17}). Furthermore, Buchin et al.~\cite{BuchinCLMMOS17} proved that if one of the curves visits any location at most a constant number of times, then the Fr\'echet distance can be computed in near linear time. Both results apply only to restricted classes of curves and hence the general case in $1$D remained open.
  
	\mypar{Our results} In this paper we settle the general question for one dimension: we give a conditional lower bound for the Fr\'echet distance between two general polygonal curves in $1$D. To do so we reduce (in linear time) from the \emph{Orthogonal Vector Problem}: given two sets of vectors, is there a pair of orthogonal vectors, one from each set? For vectors of dimension $d = \omega (\log n)$ no algorithm running in strongly subquadratic time is known. Furthermore, an algorithm with such a running time does not exist in various computational models~\cite{kane2017orthogonal} and would have far-reaching consequences~\cite{AbboudBDN18}. In particular, the existence of a strongly subquadratic algorithm for the Orthogonal Vector Problem would imply that the Strong Exponential Time Hypothesis fails. Our reduction hence implies that no strongly subquadratic algorithm for approximating the Fr\'echet distance within a factor less than~$3$ exists unless \SETH{} fails. 

	Our result also improves upon the previously best known conditional lower bound for curves in $2$D by Bringmann and Mulzer~\cite{BringmannM16} (approximation within a factor less than~$1.399$). Furthermore, we argue that similar reductions, based on a ``traditional'' encoding of the Orthogonal Vectors Problem, cannot achieve a lower bound better than 3.

	Section~\ref{sec:prelim} gives various definitions and background. In particular, we recall an asymmetric variant of the Fr\'echet distance introduced by Alt and Godau~\cite{altgodau}, the so-called \emph{partial Fr\'echet distance}. In Section~\ref{sec:results} we succinctly state all our results and in Section~\ref{sec:traditional} we briefly argue why traditional reductions cannot achieve a lower bound better than 3. In Section~\ref{sec:partial} we present our reduction to the partial Fr\'echet distance, followed in Section~\ref{sec:frechet} by the reduction to the  Fr\'echet distance. The remainder of the paper covers the two most popular variants of the Fr\'echet distance, namely the discrete Fr\'echet distance and the weak Fr\'echet distance.

	The \emph{discrete Fr\'echet distance}~\cite{agarwal2014computing,eiter1994computing} considers only distances between vertices of the curves. Bringmann and Mulzer~\cite{BringmannM16} proved that there is no strongly subquadratic time algorithm for approximating the discrete Fr\'echet distance in any dimension within a factor less than~$1.399$ unless \SETH{} fails. In Section~\ref{sec:discrete} we extend our reduction for the (regular) Fr\'echet distance to the discrete Fr\'echet distance, and hence also strengthen this lower bound to an approximation factor of~$3$.

	For the \emph{weak Fr\'echet distance}~\cite{altgodau} the reparameterizations are not required to be monotone. The missing monotonicity condition gives this variant a very different flavor than the regular and the discrete Fr\'echet distance. For the weak Fr\'echet distance only few complexity results are known: it can be computed in quadratic time~\cite{har2014frechet}, and there is an $\Omega(n \log n)$ lower bound in the algebraic computation tree model for curves in $2$D~\cite{buchin2007difficult}. The latter paper also presents a linear-time algorithm for a variant for curves in $1$D, which allows for a broader class of reparameterizations (see Section~\ref{sec:weak} for details). In Section~\ref{sec:weak} we significantly improve the lower bound by showing that there is no strongly subquadratic time algorithm for approximating the weak Fr\'echet distance within a factor less than~$3$ unless \SETH{} fails. Again we reduce from Orthogonal Vectors, but the missing monotonicity forces us to use a different reduction, which applies only to curves in two and higher dimensions. However, the same reduction can also be used for the discrete weak Fr\'echet distance for curves in~$1$D.

	This leaves the general weak Fr\'echet distance in $1$D as the only remaining case without a conditional lower bound. Interestingly the weak Fr\'echet distance in $1$D is actually computable in subquadratic time. More specifically, in Section~\ref{sec:algoWeak} we present a linear time algorithm for computing the general continuous weak Fr\'echet distance in $1$D. Our algorithm first simplifies the curves independently, removing vertices that do not influence the distance. Then a greedy strategy allows us to compute the weak Fr\'echet distance in linear time.

\section{Preliminaries}\label{sec:prelim}
	For a sequence of vertices~$p_1,\dots,p_n\in\R$, let~$\C{p_1,\dots,p_n}$ denote the continuous function~$P\from[1,n]\to\R$ defined by~$P(i+\lambda)=p_i+\lambda(p_{i+1}-p_i)$ with~$i\in\N$ and~$\lambda\in[0,1]$.
	We say that~$P$ is a one-dimensional curve on~$|P|=n$ vertices.
	Two curves~$P=\C{p_1,\dots,p_{|P|}}$ and~$Q=\C{q_1,\dots,q_{|Q|}}$ can be composed into the curve~$P\circ Q=\C{p_1,\dots,p_{|P|},q_1,\dots,q_{|Q|}}$ on~$|P|+|Q|$ vertices.
	For a natural number~$k$, let~$k\cdot P$ be the composition of~$k$ copies of~$P$.
	For~$1\leq a\leq b\leq |P|$, let~$P[a,b]$ be the curve defined by the sequence of vertices starting at~$P(a)$, followed by the sequence of~$p_i$ with~$a<i<b$, and ending at~$P(b)$.

	The \emph{Fr\'echet distance} between two curves~$P$ and~$Q$ is based on matchings between those curves.
	A \emph{matching} is a pair of functions~$\phi_1$ and~$\phi_2$ that map a time parameter~$t\in[0,1]$ to a position along~$P$ and~$Q$ respectively.
	For a continuous matching, we require that~$\phi_1\from[0,1]\to[1,|P|]$ and~$\phi_2\from[0,1]\to[1,|Q|]$ are continuous non-decreasing surjections.
	For a discrete matching, we require that~$\phi_1\from[0,1]\to\{1,|P|\}$ and~$\phi_2\from[0,1]\to\{1,|Q|\}$ are non-decreasing surjections (with a discrete range).
	For curves~$P$ and~$Q$, the \emph{width} of a matching is the maximum distance between~$P(\phi_1(t))$ and~$Q(\phi_2(t))$, defined as
	$$\width_{P,Q}(\phi_1,\phi_2)=\max_{t\in[0,1]}\|P(\phi_1(t))-Q(\phi_2(t))\|\text{.}$$
	The (continuous) \emph{Fr\'echet distance} between two curves~$P$ and~$Q$ is defined as
	$$\dF(P,Q)=\inf_{\phi_1,\phi_2}\width_{P,Q}(\phi_1,\phi_2)$$
	where~$(\phi_1,\phi_2)$ ranges over continuous matchings.
	The \emph{discrete Fr\'echet distance}~$\ddF$ is defined similarly, except that~$(\phi_1,\phi_2)$ ranges over continuous matchings.
	We also consider the following (asymmetric) variant of the Fr\'echet distance, as introduced in~\cite{altgodau}.
	A \emph{partial matching} from~$P$ to~$Q$ is a matching between~$P$ and a subcurve~$Q[a,b]$ of~$Q$.
	In the discrete case, we impose that~$a$ and~$b$ are integers.
	The \emph{partial Fr\'echet distance}~$\doF$ from~$P$ to~$Q$ is~$\doF(P,Q)=\inf_{0\leq a\leq b\leq|Q|}\dF(P,Q[a,b])$.
	The \emph{weak Fr\'echet distance} is defined in Section~\ref{sec:weak}.

	The \emph{free space diagram} is a frequently used tool for computing the Fr\'echet distance.
	For two curves~$P$ and~$Q$, the $\e$-free space is the set of pairs~$(x,y)\in [1,|P|]\times[1,|Q|]$ for which~$\|P(x)-Q(y)\|\leq\e$.
	A matching~$(\phi_1,\phi_2)$ of width~$\e$ traces a bimonotone path~$t\mapsto(\phi_1(t),\phi_2(t))$ from~$(0,0)$ to~$(|P|,|Q|)$ through the~$\e$-free space.
	Indeed, any such bimonotone path yields an~$\e$-matching.
	We tend to draw free space diagram using arc-length parameterizations of the curves on the~$x$- and~$y$-axes.

	In contrast to a matching, a \emph{cut} of width~$\e$ and complexity~$k$ is a pair~$(\Gamma_1,\Gamma_2)$ of sequences of~$k$ paths~$\Gamma_1=\{\gamma_{1,1},\dots,\gamma_{1,k}\}$ and~$\Gamma_2=\{\gamma_{2,1},\dots,\gamma_{2,k}\}$ with the following properties.
	\begin{itemize}
		\item For any~$i$, we have~$\gamma_{1,i}\from[0,1]\to[1,|P|]$ and~$\gamma_{2,i}\from[0,1]\to[1,|Q|]$.
		\item For any~$i$ and~$t$ and~$\delta>0$, the pair~$(\gamma_{1,i}(t),\gamma_{2,i}(t))$ does not lie in the~$(\e-\delta)$-free space.
		\item For any~$i<k$, we have~$\gamma_{1,i}(1)\leq\gamma_{1,i+1}(0)$ and~$\gamma_{2,i}(1)\geq\gamma_{2,i+1}(0)$.
	\end{itemize}
	We say that a cut of complexity~$k$ starts at~$(\gamma_{1,1}(0),\gamma_{2,1}(0))$ and ends at~$(\gamma_{1,k}(1),\gamma_{2,k}(1))$.
	If a cut of width~$\e$ starts on~$[1,|P|]\times\{1\}$ or~$\{|P|\}\times[1,|Q|]$ and ends on~$[1,|P|]\times\{|Q|\}$ or~$\{1\}\times[1,|Q|]$, then the (continuous) Fr\'echet distance between~$P$ and~$Q$ is at least~$\e$ \cite{buchin2015computing}.
	Similarly, if a cut of width~$\e$ starts on~$[1,|P|]\times\{1\}$ and ends on~$[1,|P|]\times\{|Q|\}$, then the partial Fr\'echet distance from~$P$ to~$Q$ is at least~$\e$.

\subsection{Orthogonal Vectors}
	Let~$U=\{\u_0,\dots,\u_{n-1}\}$ and~$V=\{\v_0,\dots,\v_{m-1}\}\subseteq\{0,1\}^d$ be sets of boolean vectors of dimension~$d$.
	The \emph{Orthogonal Vectors} problem (\OV{}) asks for~$n=m$, whether vectors~$\u\in U$ and~$\v\in V$ exist for which~$\u$ and~$\v$ are orthogonal; that is,~$\sum_{i=0}^{d-1} u_i v_i=0$.
	For any~$\delta>0$, \OV{} has no~$\bigO(n^{2-\delta}d^{\bigO(1)})$ time algorithm unless \SETH{} (and the weaker hypothesis \SETH{}'~\cite{Williams04}) fails.
	Denote by~\uOV{} the variant of \OV{} where we allow~$n\neq m$.
	For any~$\delta>0$, \uOV{} has no~$\bigO((nm)^{1-\delta}d^{\bigO(1)})$ time algorithm unless \SETH{}' fails~\cite{Bringmann14}.
	The reductions in this paper use the following restriction on~$U$ and~$V$.
	\begin{definition}[Nontrivial instance]
		Nonempty sets~$U$ and~$V\subseteq\{0,1\}^d$ for which~$d\notin\bigO(1)$ and neither~$U$ nor~$V$ contains the zero vector.
	\end{definition}

	\begin{lemma}\label{lem:nontrivial}
		If there is an algorithm~\alg{} that solves nontrivial instances in time~$T(n,m,d)$, then~\uOV{} can be solved in time~$\bigO((n+m)d+T(n,m,d))$.
	\end{lemma}
	\begin{proof}
		We can test in~$\bigO(nmd)$ time whether an instance~$(U,V)$ of~\uOV{} is nontrivial.
		If so, we return~$\alg{}(U,V)$ in time~$\bigO((n+m)d+T(n,d))$.
		Otherwise we solve it in~$\bigO((n+m)d)$ time using the following three cases.
		If~$U$ or~$V$ is empty, then there is no orthogonal pair of vectors.
		If~$d$ is at most a constant, then~$U$ and~$V$ contain at most~$2^d$ vectors, so the instance can be solved in constant time.
		If neither~$U$ and~$V$ are empty, but~$U$ or~$V$ contains the zero vector, then the zero vector is orthogonal to any vector from the other set.
	\end{proof}
	We obtain Corollaries~\ref{cor:uOV} and (using an analogous argument)~\ref{cor:OV} from Lemma~\ref{lem:nontrivial}.
	Hence, we assume~$(U,V)$ to be a nontrivial instance for the remainder of this paper.
	\begin{corollary}\label{cor:uOV}
		\SETH{}' fails if for some~$\delta>0$, there is a~$\bigO((nm)^{1-\delta}d^{\bigO(1)})$ time algorithm for nontrivial instances of~\uOV{}.
	\end{corollary}
	\begin{corollary}\label{cor:OV}
		\SETH{}' fails if for some~$\delta>0$, there is a~$\bigO(n^{2-\delta}d^{\bigO(1)})$ time algorithm for nontrivial instances of~\OV{}.
	\end{corollary}

\section{Results}\label{sec:results}
	For any polynomial restriction of~$1\leq |P|\leq |Q|$ and any~$\delta>0$, we show for several variants of the Fr\'echet distance that there is no factor~$(3-\e)$-approximation algorithm with the running times listed in Table~\ref{tab:results} unless \SETH{}' fails.
	The continuous weak Fr\'echet distance between curves in one dimension is a special case, and we give a linear-time exact algorithm.

	\begin{table}[h]%
	\caption{Asymptotic running times with no~$(3-\e)$-approximation, assuming \SETH{}' and~$|P|\leq|Q|$.
	Results listed for continuous and discrete curves in one dimension and higher dimensions.\label{tab:results}}%
	\begin{tabular}{ r r | c c c c }
		                  &             & continuous $1$D        & discrete $1$D          & continuous $k$D        & discrete $k$D         \\
		\hline
		        Fr\'echet &  $\dF(P,Q)$ & $(|P|+|Q|)^{2-\delta}$ & $(|P|+|Q|)^{2-\delta}$ & $(|P|+|Q|)^{2-\delta}$ & $(|P|+|Q|)^{2-\delta}$\\
		partial Fr\'echet & $\doF(P,Q)$ & $(|P||Q|)^{1-\delta}$  & $(|P||Q|)^{1-\delta}$  & $(|P||Q|)^{1-\delta}$  & $(|P||Q|)^{1-\delta}$ \\
		   weak Fr\'echet & $\dwF(P,Q)$ & ---                    & $(|P||Q|)^{1-\delta}$  & $(|P||Q|)^{1-\delta}$  & $(|P||Q|)^{1-\delta}$ 
	\end{tabular}%
	\end{table}

\section{Traditional Reductions}\label{sec:traditional}
	Over the past few years, several conditional lower bounds for computing the Fr\'echet distance have been found~\cite{Bringmann14,BringmannM16}. In each case, the reduction is (or can be phrased as one) from Orthogonal Vectors.
	The common pattern in these reductions is that each vector $\u\in U$ is encoded as a curve~$P_\u$, and each vector~$\v\in V$ is encoded as a curve~$Q_\v$, with the crucial property that the distance between~$P_\u$ and~$Q_\v$ is at most~$\e$ if~$\u$ and~$\v$ are orthogonal, and at least~$c\e$ otherwise (for some~$c>1$).
	We refer to a reduction that encodes vectors in this way as a \emph{traditional} reduction.
	In this paper, we give traditional reductions with~$c=3$, and in Lemma~\ref{lem:traditional} we show that traditional reductions with~$c>3$ do not exist.

	\begin{lemma}
		There is no traditional reduction with~$c>3$.\label{lem:traditional}
	\end{lemma}
	\begin{proof}
		Consider vectors~$\u_1,\v_1,\u_2$ and~$\v_2$ such that each pair of vectors except~$\u_1$ and~$\v_2$ is orthogonal. By the triangle inequality we have~$\dF(P_{\u_1},Q_{\v_2})\leq \dF(P_{\u_1},Q_{\v_1})+\dF(Q_{\v_1},P_{\u_2})+\dF(P_{\u_2},Q_{\v_2})\leq 3\e<c\e$, contradicting that~$\dF(P_{\u_1},Q_{\v_2})\geq c\e$.
	\end{proof}

\section{Partial Fr\'echet distance}\label{sec:partial}
	In this section we give a~$\bigO((n+m)d)$ time transformation from a nontrivial instance~$(U,V)$ to a pair of one-dimensional curves~$P$ and~$Q$ of sizes~$\Theta(nd)$ and~$\Theta((n+m)d)$ respectively.
	In particular, if~$n$ is small compared to~$m$, then~$P$ and~$Q$ will have an unbalanced number of vertices.
	We show that~$\doF(P,Q)\leq 1$ if~$(U,V)$ is a \YES-instance, and~$\doF(P,Q)\geq 3$ otherwise.
	Hence, for any polynomial restriction of~$1\leq |P|\leq |Q|$ and any~$\delta,\e>0$, a~$\bigO((|P||Q|)^{1-\delta})$ time~$(3-\e)$-approximation algorithm of the partial Fr\'echet distance violates \SETH'.
	Define~$P$ and~$Q$ as below.
	For a convenient analysis, we exhaustively remove vertices that lie on the segment between their neighbors so that edges have positive length and alternate in direction.
	\begin{align*}
		P_{\u\in U}       &= \C{0}\circ\bigcirc_{i=1}^d (\C{10-2u_i,4})\circ\C{0}
		                  &  \hspace{-\linewidth}P_{(0,1,0,1)}~
		                     \curve[myblue]{0,10,4,8,4,10,4,8,0}\\
		Q_{\v\in V}       &= \C{1}\circ\bigcirc_{i=1}^d (\C{9+2v_i,3})\circ\C{1}
		                  &  \hspace{-\linewidth}Q_{(0,0,1,1)}~
		                     \curve[myred]{1,9,3,9,3,11,3,11,1}\\
		P^*               &= \C{2}\circ d\cdot\C{10,4}\circ\C{2}
		                  &  \curve[myblue]{2,10,4,10,4,10,4,10,2}\\
		Q^*               &= \C{3}\circ d\cdot\C{9,5}\circ\C{3} 
		                  &  \curve[myred]{3,9,5,9,5,9,5,9,3}\\
		P^+               &= \C{4}\circ d\cdot\C{ 8,6}\circ\C{4}
		                  &  \curve[myblue]{4,8,6,8,6,8,6,8,4}\\
		Q^+               &= \C{5,7,5}
		                  &  \curve[myred]{5,7,7,7,7,7,7,7,5}\\[1ex]
		P_\mathit{sep}    &= \C{0}\circ P^*\circ\C{2}\circ P^+\circ\C{2}\circ P^+\circ P^+\circ\C{2}\circ P^+\circ\C{2}\circ P^*\circ \C{0}\\
		P_\mathit{enter}  &= (d+1)\cdot\C{4,10}\hspace{1.625em}\circ\C{2}\circ P^+\circ P^+\circ \C{2}\circ P^+\circ \C{2}\circ P^*\circ \C{0}\\
		P                 &= P_\mathit{enter}\circ \bigcirc_{i=0}^{n-2}(P_{\u_i}\circ P_\mathit{sep})\circ P_{\u_{n-1}}\circ\rev{P_\mathit{enter}}\\[1ex]
		Q_\mathit{sep}    &= \C{1}\circ Q^*\circ\C{1}\circ Q^*\circ Q^*\circ Q^*\circ Q^*\circ\C{1}\circ Q^*\circ\C{1}\\
		Q                 &= Q_\mathit{sep}\circ\bigcirc_{k=0}^{n+m-2}(Q_{\v_{k\bmod m}}\circ Q_\mathit{sep})\text{.}
	\end{align*}
	The gadget~$Q^+$ is not used to construct~$P$ and~$Q$, but will be used in a later reduction.
	Observe that matchings of width~$1$ exist for the following pairs of curves:
	$(P^+, Q^+)$,
	$(P^+, Q^*)$,
	$(P^*, Q^*)$, 
	$(P^*, Q_{\v\in V})$, and
	$(P_{\u\in U}, \C{1}\circ Q^*\circ\C{1})$.
	We will make extensive use of these matchings.

\subsection{Yes-instance}\label{sec:sat}
	Consider a nontrivial \YES-instance of \uOV{}.
	We construct a matching of width~$1$ between~$P$ and a subcurve of~$Q$.
	To define this matching, we first label various vertices.

	For~$i\in\{1,\dots,n-1\}$, let~$a_i$ and~$b_i$ respectively be the index in~$P$ of respectively the third and the fourth vertex at location~$2$ in the~$(i-1)$-th copy of~$P_\mathit{sep}$.
	Moreover, let~$a_0$ and~$b_0$ respectively be the index in~$P$ of respectively the first and second vertex at location~$2$ in~$P_\mathit{enter}$.
	Symmetrically, let~$a_n$ and~$b_n$ respectively be the index in~$P$ of respectively the second-to-last and the last vertex at location~$2$ in~$\rev{P_\mathit{enter}}$.
	Similarly, define~$s_i$ and~$t_i$ respectively to be the index in~$P$ of respectively the first and last vertex of the gadgets.

	For~$k\in\{0,\dots,n+m-1\}$, let~$c_k$ be the index in~$Q$ of the central vertex of the~$k$-th copy of~$Q_\mathit{sep}$.
	Let~$l_k$ be the index in~$Q$ of the last vertex at location~$5$ in the second copy of~$Q^*$ of the~$k$-th copy of~$Q_\mathit{sep}$.
	Symmetrically, let~$r_k$ be the index in~$Q$ of the first vertex at location~$5$ in the fifth copy of~$Q^*$ in the~$k$-th copy of~$Q_\mathit{sep}$.
	Similarly, define~$s'_k$ and~$t'_k$ respectively to be the index in~$Q$ of respectively the first and last vertex of the~$k$-th copy of~$Q_\mathit{sep}$.
	We illustrate these indices in Figure~\ref{fig:sepIndices}.

	\begin{figure}[h]\centering%
		\hspace*{\stretch{1}}%
		\curveAnnot[myblue]{
			\draw (3.7,0) -- (-.5,0) node[left=.5] {$s_0$};
			\draw (1.7,9.5) -- (-.5,9.5) node[left=.5] {$a_0$};
			\draw (1.7,25.5) -- (-.5,25.5) node[left=.5] {$b_0$};
			\node[left=.5] at (-.5,41) {$t_0$};
		}{4,10,4,10,4,10,4,10,4,10,	2,	8,6,8,6,8,6,8,	4,	8,6,8,6,8,6,8,	2,	8,6,8,6,8,6,8,	2,	10,4,10,4,10,4,10,	0}%
		\hspace*{\stretch{1}}%
		\curveAnnot[myblue]{
			\node[left=.5] at (-.5,0) {$s_i$};
			\draw (1.7,15.5) -- (-.5,15.5) node[left=.5] {$a_i$};
			\draw (1.7,31.5) -- (-.5,31.5) node[left=.5] {$b_i$};
			\node[left=.5] at (-.5,47) {$t_i$};
		}{0,	10,4,10,4,10,4,10,	2,	8,6,8,6,8,6,8,	2,	8,6,8,6,8,6,8,	4,	8,6,8,6,8,6,8,	2,	8,6,8,6,8,6,8,	2,	10,4,10,4,10,4,10,	0}%
		\hspace*{\stretch{1}}%
		\curveAnnot[myred]{
			\draw (0.7,0) -- (-.5,0) node[left=.5] {$s'_k$};
			\draw (4.7,13.5) -- (-.5,13.5) node[left=.5] {$l_k$};
			\draw (2.7,23.5) -- (-.5,23.5) node[left=.5] {$c_k$};
			\draw (4.7,33.5) -- (-.5,33.5) node[left=.5] {$r_k$};
			\draw (0.7,47) -- (-.5,47) node[left=.5] {$t'_k$};
		}{1,	9,5,9,5,9,5,9,	1,	9,5,9,5,9,5,9,	3,9,5,9,5,9,5,9,	3,9,5,9,5,9,5,9,	3,9,5,9,5,9,5,9,	1,	9,5,9,5,9,5,9,	1}%
		\hspace*{\stretch{1}}%
		\caption{$P_\mathit{enter}$ (left), $P_\mathit{sep}$ (middle), and $Q_\mathit{sep}$ (right) for~$d=4$.\label{fig:sepIndices}}
	\end{figure}
	
	\noindent
	For two curves of equal size, call~$(\phi_1,\phi_2)$ a \emph{synchronous matching} if~$\phi_1=\phi_2$.
	If~$\u$ and~$\v$ are orthogonal~$d$-dimensional boolean vectors, then the synchronous matching between~$P_\u$ and~$Q_\v$ has width~$1$, so~$\dF(P_\u,Q_\v)\leq 1$.
	Since~$U,V$ is a \YES-instance, we can pick~$i^*$ and~$j^*$ such that~$\u_{i^*}\in U$ and~$\v_{j^*}\in V$ are orthogonal vectors.
	Let~$h^*=(j^*-i^*) \bmod m$ and~$k^*=h^*+i^*$.
	As depicted schematically in Figure~\ref{fig:sketchSmall}, we construct a matching of width~$1$ between~$P$ and~$Q[l_{h^*},r_{h^*+n}]$.
	This matching is composed of three matchings of width~$1$ between the following pairs of curves:~$(P[1,a_{i^*}],Q[l_{h^*},c_{k^*}])$,~$(P[a_{i^*},b_{i^*+1}],Q[c_{k^*},c_{k^*+1}])$ and~$(P[b_{i^*+1},|P|],Q[c_{k^*+1},r_{h^*+n}])$.
	These matchings are constructed in Lemmas~\ref{lem:paths} and~\ref{lem:pair}.
	\begin{figure}[h]\centering%
		\hspace*{\stretch{1}}%
		\includegraphics{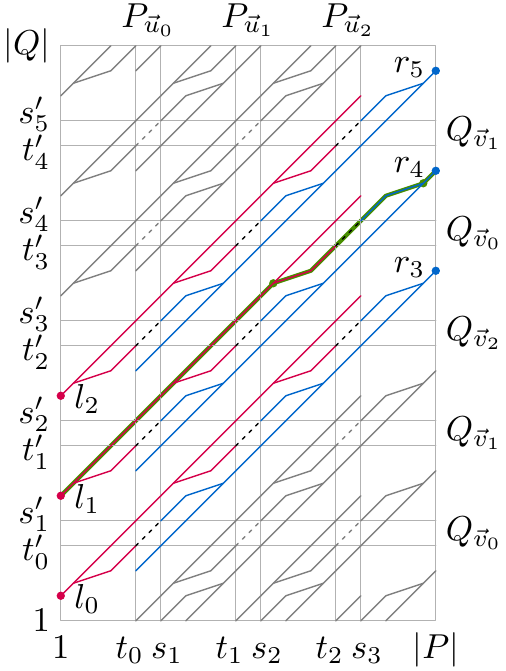}%
		\hspace*{\stretch{1}}%
		\includegraphics{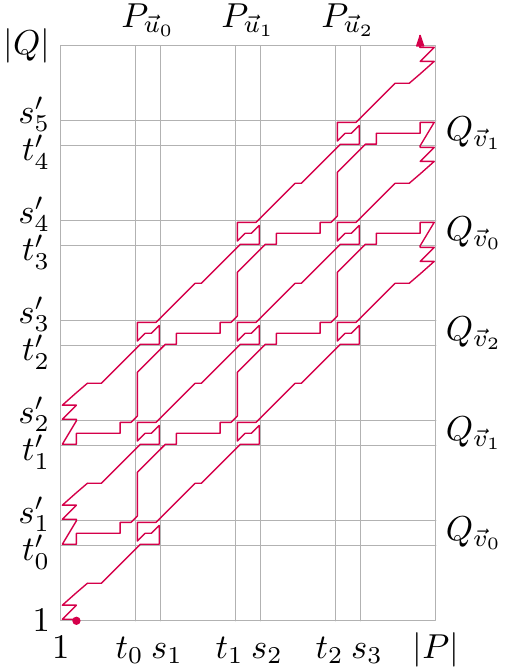}%
		\hspace*{\stretch{1}}%
		\caption{Left: relevant paths in the free space with~$n=m=3$. Dotted: potential matchings between~$P_{\u_i}$ and~$Q_{\v_{k\bmod m}}$. Highlighted: a matching of~$P$ and~$Q[l_1,r_4]$ if~$\u_2$ and~$\v_0$ are orthogonal. Right: a cut for a \NO-instance (schematically).\label{fig:sketchSmall}}%
	\end{figure}

	\begin{lemma}
		Let~$i\in\{0,\dots,n-1\}$ and~$h\in\{0,\dots,m-1\}$, then~$\dF(P[1,a_i],Q[l_h,c_{h+i}])\leq 1$ and~$\dF(P[b_{i+1},|P|],Q[c_{h+i+1},r_{h+n}])\leq 1$.\label{lem:paths}
	\end{lemma}
	\begin{proof}
		Observe that the synchronous matching has width~$1$ for the following pairs of curves.
		Cases~\enumi{1} and~\enumi{2} as well as cases~\enumi{3} and~\enumi{4} are symmetric.
		\begin{enumerate}
			\item $P[1,a_0]=(d+1)\cdot\C{4,10}\circ\C{2}$ and~$Q[l_h,c_h]=\C{5,9,3}\circ (d-1)\cdot\C{9,5}\circ\C{9,3}$;
			\item $P[b_n,|P|]$ and~$Q[c_{h+n},r_{h+n}]$;
			\item $P[a_{i'},a_{i'+1}]$ and~$Q[c_{h+i'},c_{h+i'+1}]$ with~$i'\in\{0,\dots,i\}$.
			\item $P[b_{i'},\hspace{2pt}b_{i'+1}]$ and~$Q[c_{h+i'},c_{h+i'+1}]$ with~$i'\in\{i+1,\dots,n-1\}$.
		\end{enumerate}
		These matchings can be concatenated to obtain a matching of width~$1$ between~$P[1,a_{i}]$ and~$Q[l_{h},c_{h+i}]$, and between~$P[b_{i+1},|P|]$ and~$Q[c_{h+i+1},r_{h+n}]$.
	\end{proof}

	\begin{lemma}
		If~$\u_{i^*}$ and~$\v_{j^*}$ are orthogonal, then a matching of width~$1$ between~$P[a_{i^*},b_{i^*+1}]$ and~$Q[c_{k^*},c_{k^*+1}]$ exists for~$k^*=h^*+i^*$ with~$h^*=(j^*-i^*) \bmod m$.\label{lem:pair}
	\end{lemma}
	\begin{proof}
		We have~$P[a_{i^*},b_{i^*+1}]=P[a_{i^*},t_{i^*}]\circ P_{\u_{i^*}}\circ P[t_{i^*+1},b_{i^*+1}]$.
		Since~$j^*=k^* \bmod m$, we have~$Q[c_{k^*},c_{k^*+1}]=Q[c_{k^*},t'_{k^*}]\circ Q_{\v_{j^*}}\circ Q[s'_{k^*+1},c_{k^*+1}]$.
		Since~$\u_{i^*}$ and~$\v_{j^*}$ are orthogonal, the synchronous matching between~$P_{\u_{i^*}}$ and~$Q_{\v_{j^*}}$ has width~$1$.
		It remains to show that there exists a matching of width~$1$ between (a)~$P[a_{i^*},t_{i^*}]$ and~$Q[c_{k^*},t'_{k^*}]$, and between (b)~$P[t_{i^*+1},b_{i^*+1}]$ and~$Q[s'_{k^*+1},c_{k^*+1}]$.
		We show case (a), the other case is symmetric.
		We have~$P[a_{i^*},t_{i^*}]=P[a_{i^*},b_{i^*}]\circ P[b_{i^*},t_{i^*}]$ and~$Q[c_{k^*},t'_{k^*}]=Q^*\circ Q^*\circ\C{1}\circ Q^*\circ\C{1}$.
		The synchronous matching between~$P[b_{i^*},t_{i^*}]$ and~$Q^*\circ\C{1}\circ Q^*\circ\C{1}$ has width~$1$, so it suffices to construct a matching between~$P[a_{i^*},b_{i^*}]=\C{2}\circ P^+\circ P^+\circ \C{2}$ and~$Q^*$.
		We illustrate such a matching in Figure~\ref{fig:fork}.
		Recall that~$d\notin\bigO(1)$, so assume that~$d\geq 2$.
		We can view~$P[a_{i^*},b_{i^*}]$ as~$(\C{2}\circ (d-1)\cdot\C{8,6})\circ(\C{6,8})\circ(\C{8,4,8,6})\circ((d-1)\cdot\C{6,8})\circ(\C{8,2})$ and~$Q^*$ as~$(\C{3}\circ (d-2)\cdot\C{9,5}\circ\C{9,7})\circ(\C{7})\circ(\C{7,5,9,7})\circ(\C{7})\circ(\C{7,3})$.
		The second and fourth parenthesized terms~$\C{6,8}$ and~$(d-1)\cdot\C{6,8}$ in this view of~$P[a_{i^*},b_{i^*}]$ can be matched to the second and fourth terms~$\C{7}$ of this view of~$Q^*$ with width~$1$.
		Moreover, we can match the first, third, and fifth pairs of parenthesized terms in these views synchronously with width~$1$.
		Thus we obtain a matching of width~$1$ between~$P[a_{i^*},b_{i^*}]$ and~$Q^*$.
	\end{proof}

	\begin{figure}\centering%
		\includegraphics{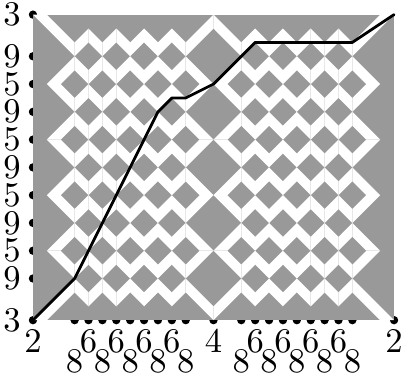}%
		\caption{$1$-free space (white) and a matching of width~$1$ for~$\C{2}\circ P^+\circ P^+\circ \C{2}$ and~$Q^*$ with~$d=5$.\label{fig:fork}}%
	\end{figure}

	\begin{lemma}
		If~$(U,V)$ is a nontrivial \YES-instance of \OV{}, then there is a matching of width~$1$ between~$P$ and~$Q[l_h,r_{h+n}]$ for some~$h\in\{0,\dots,m-1\}$.\label{lem:doFmatch}
	\end{lemma}

	\begin{corollary}
		If~$(U,V)$ is a nontrivial \YES-instance of \OV{}, then~$\doF(P,Q)\leq 1$.\label{cor:doFyes}
	\end{corollary}

\subsection{No-instance}\label{sec:unsat}
	Consider a \NO-instance of \uOV{}.
	To show that~$\doF(P,Q)\geq 3$ we construct of cut of width~$3$ between~$P$ and~$Q$ that starts on~$[1,|P|]\times\{1\}$ and ends on~$[1,|P|]\times\{|Q|\}$.
	Our cut consists of the following types of elementary pieces.
	\begin{enumerate}
		\item If all points of~$Q[q_j,q_{j'}]$ are at distance at least~$\e$ from~$p_i$, then there are cuts of width~$\e$ from~$(i,j)$ to~$(i,j')$ and from~$(i,j')$ to~$(i,j)$.
		\item Symmetrically, if all points of~$P[p_i,p_{i'}]$ are at distance at least~$\e$ from~$q_j$, then there are cuts of width~$\e$ from~$(i,j)$ to~$(i',j)$ and from~$(i',j)$ to~$(i,j)$.
		\item If~$i<i'$, there is a cut of width~$\e$ from~$(i,j)$ to~$(i',j+1)$ if~$|p_i-q_j|\geq\e$ and~$|p_{i'}-q_{j+1}|\geq\e$ and either~$p_i+2\e\leq p_{i'}$ and~$q_j\geq q_{j+1}$ or~$p_i-2\e\geq p_{i'}$ and~$q_j\leq q_{j+1}$.
		\item Symmetrically if~$j<j'$, there is a cut of width~$\e$ from~$(i+1,j')$ to~$(i,j)$ if~$|p_i-q_j|\geq\e$ and~$|p_{i+1}-q_{j'}|\geq\e$ and either~$p_i\geq p_{i+1}$ and~$q_j+2\e\leq q_{j'}$ or~$p_i\leq p_{i+1}$ and~$q_j-2\e\geq q_{j'}$.
	\end{enumerate}
	Types \enumi{1} and \enumi{2} trivially provide cuts consisting of a single straight path.
	The cuts of types \enumi{3} and \enumi{4} are constructed in Lemma~\ref{lem:stepCut}.
	Lemma~\ref{lem:cutVec} uses these to cut between~$P_\u$ and~$Q_\v$.
	\begin{lemma}
		If~$i<i'$, there is a cut of width~$\e$ from~$(i,j)$ to~$(i',j+1)$ if~$|p_i-q_j|\geq\e$ and~$|p_{i'}-q_{j+1}|\geq\e$ and either~$p_i+2\e\leq p_{i'}$ and~$q_j\geq q_{j+1}$ or~$p_i-2\e\geq p_{i'}$ and~$q_j\leq q_{j+1}$.
		\label{lem:stepCut}
	\end{lemma}
	\begin{proof}
		Let~$|p_i-q_j|\geq\e$ and~$|p_{i'}-q_{j+1}|\geq\e$ and suppose that~$p_i+2\e\leq p_{i'}$ and~$q_j\geq q_{j+1}$ (the other case where~$p_i-2\e\geq p_{i'}$ and~$q_j\leq q_{j+1}$ is symmetric).
		Let~$y$ be the minimum value in~$[j,j+1]$ for which~$|p_i-Q(y)|=\e$, or~$y=j+1$ if there is no such value.
		Then there is a straight cut of width~$\e$ from~$(i,j)$ to~$(i,y)$.
		We show that~$p_{i'}$ is at distance at least~$\e$ from any point~$Q(y')$ with~$y'\in[y,j+1]$.
		Indeed,~$Q(y')\leq Q(y)\leq p_i+\e\leq p_{i'}-\e$.
		Hence, the straight cut from~$(i',y)$ to~$(i',j+1)$ also has width at least~$\e$, so the composition of the straight cuts yields a cut of width~$\e$ from~$(i,j)$ to~$(i',j+1)$.
	\end{proof}

	\begin{lemma}
		If~$\u$ and~$\v$ are not orthogonal, then there is a cut of width~$3$ between~$P_\u$ and~$Q_\v$ that starts at~$(|P_\u|-1,1)$ and ends at~$(2,|Q_\v|)$.\label{lem:cutVec}
	\end{lemma}
	\begin{proof}
		Since~$\u$ and~$\v$ are not orthogonal, $u_b=v_b=1$ for some index~$b\in\{0,\dots,d-1\}$.
		Compose a cut of width~$3$ (see also Figure~\ref{fig:cutVec}) out of these elementary pieces:
		\enumi{3} from~$(|P_\u|-1,1)$ to~$(|P_\u|,2)$; \enumi{2} to~$(|P_\u|,|Q_\v|-1)=(2d+1,2d)$; $2(d-1-b)$ copies of~\enumi{4} to~$(2b+3,2b+2)$; \enumi{1} to~$(2b+1,2b+2)$; $2b$ copies of~\enumi{4} to~$(1,2)$; \enumi{2} to~$(1,|Q_\v|-1)$; and~\enumi{3} to~$(2,|Q_\v|)$.
	\end{proof}
	\begin{figure}\centering%
		\includegraphics{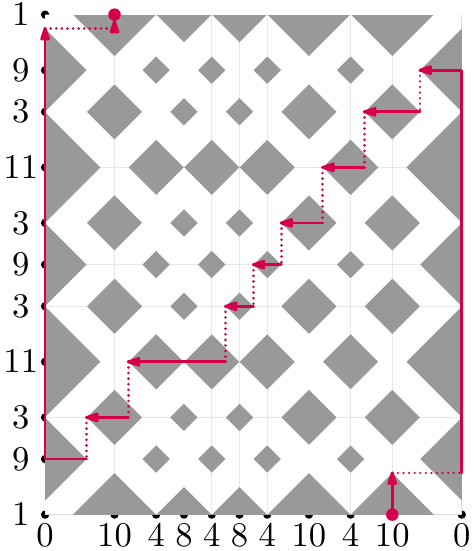}%
		\caption{$3$-free space with a cut of width~$3$ for~$P_{(0,1,1,0,0)}$ and~$Q_{(0,1,0,1,0)}$.\label{fig:cutVec}}%
	\end{figure}

	\begin{figure}\centering%
		\includegraphics[scale=.78]{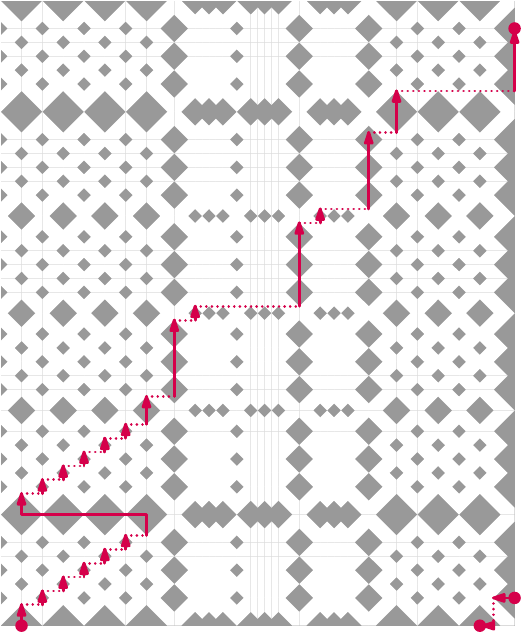}%
		\includegraphics[scale=.78]{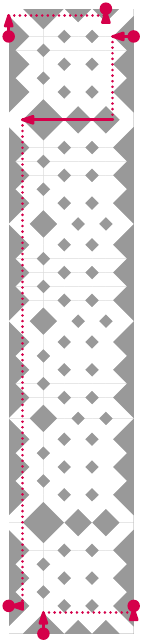}%
		\includegraphics[scale=.78]{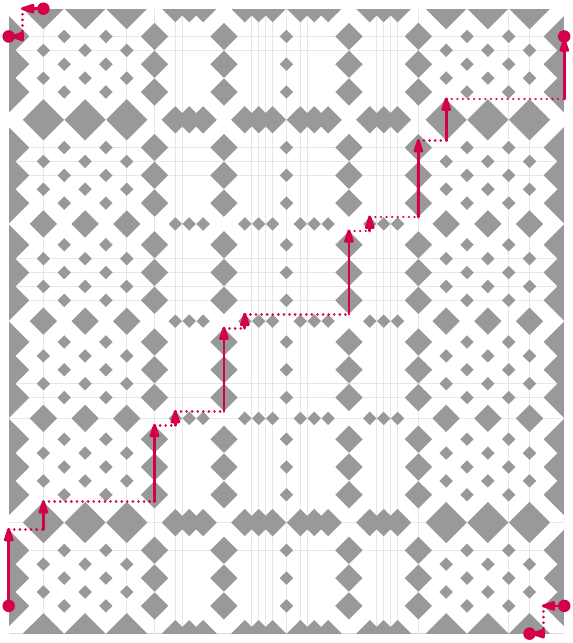}%
		\includegraphics[scale=.78]{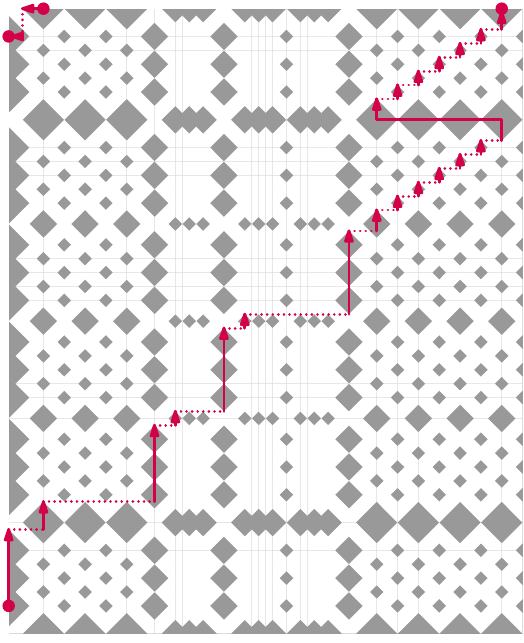}\\%
		\includegraphics[scale=.78]{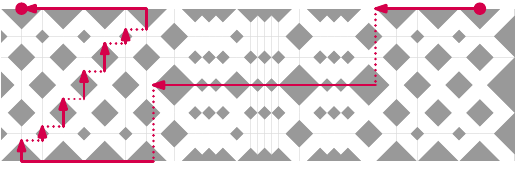}%
		\includegraphics[scale=.78]{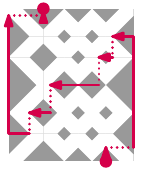}%
		\includegraphics[scale=.78]{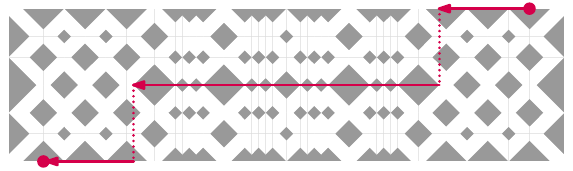}%
		\includegraphics[scale=.78]{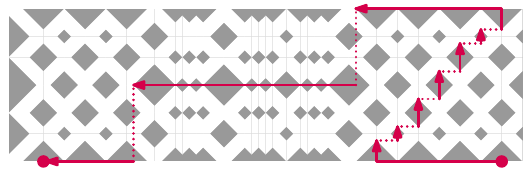}%
		\caption{The~$3$-free space with cuts of width~$3$ for (from left to right) $P_\mathit{enter}$,~$P_{(0,1,1)}$,~$P_\mathit{sep}$, or~$\rev{P_\mathit{enter}}$ and (top row)~$Q_\mathit{sep}$ or (bottom row)~$Q_{(0,1,0)}$.\label{fig:cuts}}%
	\end{figure}
	\noindent
	For the purpose of disambiguation, we will refer to the subcurve~$Q_{\v_{k\bmod m}}$ of~$Q$ simply as~$Q_{\v_k}$.
	The cuts given by Lemma~\ref{lem:cutVec} are connected as follows, as illustrated in Figures~\ref{fig:sketchSmall} and~\ref{fig:cuts}.
	\begin{enumerate}
	\item[\enumi{a}] For~$0\leq h\leq m-1$ and~$0\leq i\leq n-2$, we cut from the end of the cut between~$P_{\u_i}$ and~$Q_{\v_{h+i}}$ to the start of the cut between~$P_{\u_{i+1}}$ and~$Q_{\v_{h+i+1}}$.
	\item[\enumi{b}] Furthermore, we cut from~$(2,s'_h)$ to the start of the cut between~$P_{\u_{0}}$ and~$Q_{\v_h}$.
	\item[\enumi{c}] Similarly, we cut from the end of the cut between~$P_{\u_{n-1}}$ and~$Q_{\v_{h+n-1}}$ to~$(|P|-1,t'_{h+n})$.
	\item[\enumi{d}] Finally, for~$0\leq h\leq m-2$, we cut from~$(|P|-1,t'_{h+n})$ to~$(2,s'_{h+1})$.
	\end{enumerate}
	Composing these cuts yields a cut of width~$3$ from~$(2,1)$ to~$(|P|-1,|Q|)$.

	We believe that the illustrations of Figure~\ref{fig:cuts} are more helpful than the formal definitions of such cuts.
	The cuts of types~\enumi{a}, \enumi{b}, and~\enumi{c} start or end with the cut illustrated in the top and bottom of the second column of the top row of Figure~\ref{fig:cuts}.
	The remainder of the cuts of type~\enumi{a} is illustrated as the central cut in the third column of the top row.
	The cuts of type~\enumi{b} start as illustrated in the first column of the top row.
	Similarly, the cuts of type~\enumi{c} end as illustrated in the last column of the top row.
	The cuts of type~\enumi{d} are more complicated and start with the last column of the bottom row.
	Ignoring small cuts in corners, this cut is followed by the central cut of the second column of the top row, and the cut in the third column of the bottom row, repeated~$n-1$ times, followed by a final copy of the central cut of the second column of the top row and the cut of the first column of the bottom row.

	Whereas it should be evident why the cuts in the top row exist, this may not be clear for the cuts in the bottom row.
	In particular, a central elementary piece of type~\enumi{1} exists only if the corresponding vector~$\v_{h+i}$ contains a one.
	However, this is the case since all vectors are nonzero, since our instance is nontrivial.
	
	\begin{lemma}
		If~$(U,V)$ is a nontrivial \NO-instance of \OV{}, then there is a cut of width~$3$ from~$(2,1)$ to~$(|P|-1,|Q|)$.\label{lem:doFcut}
	\end{lemma}
	\begin{corollary}
		If~$(U,V)$ is a nontrivial \NO-instance of \OV{}, then~$\doF(P,Q)\geq 3$.\label{cor:doFno}
	\end{corollary}

	\noindent
	Combining Corollaries~\ref{cor:doFyes} and~\ref{cor:doFno}, we obtain Theorem~\ref{thm:doF}.

	\begin{theorem}
		For any polynomial restriction of~$1\leq |P|\leq|Q|$, the partial Fr\'echet distance from~$P$ to~$Q$ has no~$\bigO((|P||Q|)^{1-\delta})$ time~$(3-\e)$-approximation unless \SETH' fails.\label{thm:doF}
	\end{theorem}

\section{Fr\'echet distance}\label{sec:frechet}
	We use~$P$ and~$Q$ to construct two curves~$P'$ and~$Q'$ of size~$\bigO((n+m)d)$ as follows.
	\begin{align*}
		P_\mathit{skip1}  &= \C{6,4,6}\circ P^+ \circ\C{6,4,6}\\
		P_\mathit{skip2}  &= P^+\circ P^+\circ P^+\circ\C{2}\circ P^+\circ P^*\circ P^+\circ\C{2}\circ P^+\\
		P_\mathit{start}  &= \C{6}\circ (m-1)\cdot P_\mathit{skip1}\circ m\cdot P_\mathit{skip2}\\
		P'                &= P_\mathit{start}\circ P\circ\rev{P_\mathit{start}}\\[2ex]
		Q_\mathit{skip1}  &= Q^+ = \C{5,7,5}\\
		Q_\mathit{skip2}  &= Q^+\circ Q^+\circ \C{7,3,7}\circ Q^*\circ\C{7,3,7}\\
		Q_\mathit{skip3}  &= d\cdot\C{11,3}\circ\C{1}\\
		Q_\mathit{start}  &= (m-1)\cdot Q_\mathit{skip1}\circ m\cdot Q_\mathit{skip2} \circ Q_\mathit{skip3}\\
		Q'                &= Q_\mathit{start}\circ Q\circ\rev{Q_\mathit{start}}\text{.}
	\end{align*}

	\begin{figure}[b!]\centering%
		\includegraphics{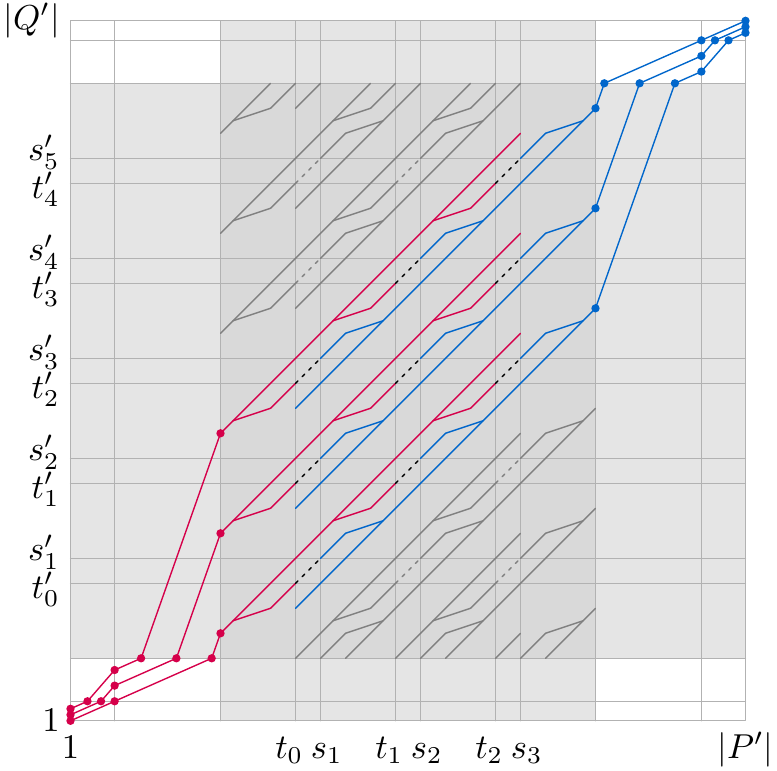}%
		\caption{Relevant paths in the free space with~$n=m=3$. The free space of~$P$ and~$Q$ shaded.\label{fig:sketch}}%
	\end{figure}

	\noindent
	We show that~$\dF(P',Q')\leq 1$ if the nontrivial instance~$(U,V)$ is a \YES-instance, and that~$\dF(P',Q')\geq 3$ otherwise (see also Figures~\ref{fig:fdSat} and~\ref{fig:fdUnsat}).
	Hence, a~$\bigO((|P|+|Q|)^{2-\delta})$ time~$(3-\e)$-approximation algorithm (with~$\e,\delta>0$) for the Fr\'echet distance violates \SETH'.

	Consider a nontrivial \YES-instance.
	Let~$l'_h$ and~$r'_h$ be the indices in~$Q'$ of respectively the~$l_h$-th and~$r_h$-th vertices of the copy of~$Q$ in~$Q'$.
	For each~$h\in\{0,\dots,m-1\}$, we construct a matching of width~$1$ between~$P_\mathit{start}$ and~$Q'[1,l'_h]$, and between~$\rev{P_\mathit{start}}$ and~$Q'[r'_{h+n},|Q'|]$, see Figure~\ref{fig:sketch}.
	It then follows from Corollary~\ref{lem:doFmatch} that~$\dF(P',Q')\leq 1$.

	We construct the matching between~$P_\mathit{start}$ and~$Q'[1,l'_h]$, the other case is symmetric.
	Match~$\C{6}$ with the first~$h$ copies of~$Q_\mathit{skip1}$.
	Match~$m-1-h$ copies of~$P_\mathit{skip1}$ with the remaining~$m-1-h$ copies of~$Q_\mathit{skip1}$.
	Match the remaining~$h$ copies of~$P_\mathit{skip1}$ with~$h$ copies of~$Q_\mathit{skip2}$.
	Match~$m-h-1$ copies of~$P_\mathit{skip2}$ with~$m-h-1$ copies of~$Q_\mathit{skip2}$.
	Match the next copy of~$P_\mathit{skip2}$ with the remainder~$Q_\mathit{skip2}\circ Q_\mathit{skip3}$ of~$Q_\mathit{start}$.
	Finally, match the remainder~$P^+\circ\C{2}\circ P^+\circ h\cdot P_\mathit{skip2}$ of~$P_\mathit{start}$ with~$Q[1,l_h]$.
	
	\begin{corollary}
		If~$(U,V)$ is a nontrivial \YES-instance of \OV{}, then~$\dF(P',Q')\leq 1$.\label{cor:dFyes}
	\end{corollary}

	\noindent	
	Now consider a nontrivial \NO-instance.
	Let~$a$ be the index in~$P_\mathit{start}$ of the last vertex at position~$10$ of the first occurrence of~$P_\mathit{skip2}$.
	We construct a cut of width~$3$ from a point on~$(a,1)$ to the start of the copy of the cut given by Corollary~\ref{lem:doFcut}.
	Similarly, we can construct a cut of width~$3$ from the end of that cut to~$(|P'|-a+1,|Q'|)$.
	We show how to construct the first cut, the other cut is symmetric.

	Let~$a'$ be the index in~$P_\mathit{start}$ of the last vertex (at position~$2$) of~$P^*$ of the last occurrence of~$P_\mathit{skip2}$.
	Consider the last two vertices of~$P^*$, namely those at positions~$10$ and~$2$, respectively.
	Any point on~$Q_\mathit{skip1}$ has distance at least~$3$ to the vertex at position~$10$.
	Similarly, for~$Q_\mathit{skip2}$ all vertices except the interior vertices of~$Q^*$ have distance at least~$3$ to the vertex at position~$10$.
	The interior vertices of~$Q^*$ have distance at least~$3$ to the vertex of~$P^*$ at position~$2$.
	Let~$b'$ be index in~$Q_\mathit{start}$ of the last vertex of~$Q^*$ at position~$9$ in the last occurrence of~$Q_\mathit{skip2}$.
	We obtain a cut of width~$3$ from~$(a,1)$ to~$(a',b')$.
	Let~$b"=b'+1$ and let~$a"$ be the index in~$P'$ of the second vertex (at location~$10$) of~$P$.
	There is a type~\enumi{3} cut of width~$3$ from~$(a',b')$ to~$(a",b")$.
	Finally, we construct a cut of width~$3$ between~$d\cdot\C{4,10}$ and~$\C{3,7,3,7}\circ Q_\mathit{skip3}$.
	The cut starts at~$(a",b")$ and uses a cut of type~\enumi{1} followed by~$d$ cuts of type~\enumi{3} and one cut of type~\enumi{2} to reach the start of the cut given by Corollary~\ref{lem:doFcut}.
	
	\begin{corollary}
		If~$(U,V)$ is a nontrivial \NO-instance of \OV{}, then~$\dF(P',Q')\geq 3$.\label{cor:dFno}
	\end{corollary}

	\noindent
	Theorem~\ref{thm:dF} follows from Corollaries~\ref{cor:dFyes} and~\ref{cor:dFno}.
	\begin{theorem}
		The Fr\'echet distance between one-dimensional curves~$P$ and~$Q$ has no $\bigO((|P|+|Q|)^{2-\delta})$ time~$(3-\e)$-approximation unless \SETH' fails.\label{thm:dF}
	\end{theorem}

\section{Discrete Fr\'echet distance}\label{sec:discrete}
	The previous constructions can easily be adapted to show that the discrete Fr\'echet distance cannot be approximated better than a factor~$3$ in strongly subquadratic time.
	We adapt the constructed curves by introducing a constant number of vertices along each edge.

	Higher-dimensional curves~$P'$ and~$Q'$ generally have~$\bigO(|P'|^2|Q'|+|P'||Q'|^2)$ critical values.
	The Fr\'echet distance between~$P'$ and~$Q'$ is always one of the critical values~\cite{altgodau}.
	In contrast to curves in higher dimensions, where a critical value can depend on three vertices, critical values for curves in one dimension depend only on two vertices.
	In particular, for curves in one dimension, a critical value is either half the distance between two vertices of the same curve, or the distance between two vertices of different curves.
	If there are only~$c$ distinct coordinates, this means there are~$\bigO(c^2)$ critical values.
	Lemma~\ref{lem:critical} transforms curves into curves that are~$\bigO(c^2)$ times as large, such that their discrete Fr\'echet distance is the Fr\'echet distance of the original curves.
	The curves in our construction have only a constant number of distinct coordinates, leading to Corollaries~\ref{cor:discrete} and~\ref{cor:discretePartial}.

	\begin{lemma}
		For continuous one-dimensional curves~$P$ and~$Q$ with~$c$ distinct coordinates, there are curves~$P'$ and~$Q'$ of sizes~$\bigO(c^2|P|)$ and~$\bigO(c^2|Q|)$ with~$\dF(P,Q)=\ddF(P',Q')$.\label{lem:critical}
	\end{lemma}
	\begin{proof}
		Let~$X_\e$ be the set of~$\bigO(c)$ coordinates that lie at distance~$\e$ from a vertex of~$P$ or~$Q$.
		Let~$P_X$ and~$Q_X$ be copies of~$P$ and~$Q$ for which each edge is subdivided by introducing vertices at the points of~$X$ on that edge.
		The curves~$P_X$ and~$Q_X$ have~$\bigO(|X||P|)$ and~$\bigO(|X||Q|)$ vertices respectively.
		Consider a matching between~$P$ and~$Q$ of width~$\e$.
		Then the discrete Fr\'echet distance between~$P_{X_\e}$ and~$Q_{X_\e}$ is at most~$\e$.
		Let~$X$ be the union of~$X_\e$ for all critical values~$\e$.
		Then~$|X|=\bigO(c^2)$.
		Consider the curves~$P_X$ and~$Q_X$ of size~$\bigO(c^2|P|)$ and~$\bigO(c^2|Q|)$ respectively.
		Then the discrete Fr\'echet distance between~$P_X$ and~$Q_X$ is at most the Fr\'echet distance between~$P$ and~$Q$.
		Since the Fr\'echet distance is a lower bound for the discrete Fr\'echet distance, we have~$\dF(P,Q)=\ddF(P_X,Q_X)$.
	\end{proof}

	\begin{remark}
		In our construction, we can say something more: because the vertices of~$P'$ all have odd coordinates and the vertices of~$Q'$ all have even coordinates, the critical values are all integer.
		Moreover, since all vertices lie in the range~$[0,11]$, the critical values of~$P'$ and~$Q'$ are integers between~$0$ and~$11$.
	\end{remark}

	\begin{corollary}
		The discrete Fr\'echet distance between one-dimensional curves~$P$ and~$Q$ has no~$\bigO((|P|+|Q|)^{2-\delta})$ time~$(3-\e)$-approximation unless \SETH' fails.\label{cor:discrete}
	\end{corollary}

	\begin{corollary}
		For any polynomial restriction of~$1\leq |P|\leq|Q|$, the partial discrete Fr\'echet distance from~$P$ to~$Q$ has no~$\bigO((|P||Q|)^{1-\delta})$ time~$(3-\e)$-approximation unless \SETH' fails.\label{cor:discretePartial}
	\end{corollary}

\section{Weak Fr\'echet distance}\label{sec:weak}
	In this section we consider the weak Fr\'echet distance.
	The width of a path~$\Phi\subseteq[1,|P|]\times[1,|Q|]$ is~$\max_{(i,j)\in\Phi}\|P(i)-Q(j)\|$.
	A (continuous) \emph{weak Fr\'echet matching} between~$P$ and~$Q$ is a path~$\Phi\subseteq[1,|P|]\times[1,|Q|]$ that starts at~$(1,1)$ and ends at~$(|P|,|Q|)$.
	The (continuous) \emph{weak Fr\'echet distance}~$\dwF(P,Q)$ between~$P$ and~$Q$ is the minimum width over all such matchings.
	In related work~\cite{buchin2007difficult}, a variant of the weak Fr\'echet distance which we will refer to as the \emph{weak Fr\'echet distance without endpoint restrictions}~$\dwwF(P,Q)$ was considered.
	This distance is defined analogously, except that we require~$\{i\mid (i,j)\in\Phi\}=[1,|P|]$ and~$\{j\mid (i,j)\in\Phi\}=[1,|Q|]$, and not that the path~$\Phi$ starts at~$(1,1)$ and ends at~$(|P|,|Q|)$.

	We define the discrete weak Fr\'echet distance analogously, but for discrete matchings.
	Consider the graph with vertices~$\{1,\dots,|P|\}\times\{1,\dots,|Q|\}$ and edges between pairs of vertices at~$\ell^\infty$ distance~$1$, such that vertex~$(i,j)$ has (undirected) edges to~$(i,j+1)$,~$(i+1,j-1)$,~$(i+1,j)$, and~$(i+1,j+1)$.
	A \emph{discrete weak Fr\'echet matching without endpoint restrictions}~$\Phi$ between~$P$ and~$Q$ consists of the set of vertices of a path in this graph, with the requirement that~$\{i\mid (i,j)\in\Phi\}=\{1,\dots,|P|\}$ and~$\{j\mid (i,j)\in\Phi\}=\{1,\dots,|Q|\}$.
	For a \emph{discrete weak Fr\'echet matching}, this path starts at~$(1,1)$ and ends at~$(|P|,|Q|)$.

\subsection{Discrete or higher-dimensional weak Fr\'echet distance}\label{sec:reductionWeak}
	Our lower bound constructions for the weak Fr\'echet distance are similar to the one by Bringmann~\cite{Bringmann14}.
	For a nontrivial instance~$(U,V)$ of~\uOV{}, we construct the following discrete curves~$P$ and~$Q$ in one dimension:
	\begin{align*}
		P_{\u\in U}       &= \bigcirc_{i=1}^d \C{6i+2-2u_i}\\
		Q_{\v\in V}       &= \bigcirc_{i=1}^d \C{6i+1+2v_i}\\[1ex]
		P_\mathit{skip}   &= \C{3}\circ P_{\vec{0}}\circ\C{6d+9}\\
		P                 &= \C{0}\circ P_\mathit{skip}\circ\rev{P_{\vec{1}}}\circ\bigcirc_{i=0}^{n-1}(P_{\u_i}\circ\rev{P_{\vec{1}}})\circ P_\mathit{skip}\circ\C{6d+12}\\
		Q                 &= \C{0,3}\circ\\
		&\hphantom{=~~}      Q_{\vec{1}}\circ\bigcirc_{j=0}^{m-1}(\C{6d+9}\circ\rev{Q_{\vec{0}}}\circ Q_{\v_j}\circ\rev{Q_{\vec{0}}}\circ\C{3}\circ Q_{\vec{1}})~\circ\\
		&\hphantom{=~~}      \C{6d+9,6d+12}\text{.}
	\end{align*}
	\noindent
	Alternatively, we construct the following continuous curves~$P$ and~$Q$ in two dimensions:
	\begin{align*}
		P_{\u\in U}       &= \bigcirc_{i=1}^d \C{(6i,1),(6i,2u_i),(6i+6,2u_i),(6i+6,1)}
		                  &  \hspace{-\linewidth}\includegraphics{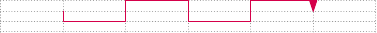}\\
		Q_{\v\in V}       &= \bigcirc_{i=1}^d \C{(6i,0),(6i,1-2v_i),(6i+6,1-2v_i),(6i+6,0)}
		                  &  \hspace{-\linewidth}\includegraphics{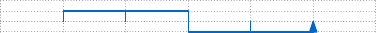}\\[1ex]
		P_\mathit{skip}   &= \C{(3,1)}\circ P_{\vec{0}}\circ\C{(6d+9,1)}\\
		P                 &= \C{(0,1)}\circ P_\mathit{skip}\circ\rev{P_{\vec{1}}}\circ\bigcirc_{i=0}^{n-1}(P_{\u_i}\circ\rev{P_{\vec{1}}})\circ P_\mathit{skip}\circ\C{(6d+12,1)}\\
		Q                 &= \C{(0,0),(3,0)}\circ\\
		&\hphantom{=~~}      Q_{\vec{1}}\circ\bigcirc_{j=0}^{m-1}(\C{(6d+9,0)}\circ\rev{Q_{\vec{0}}}\circ Q_{\v_j}\circ\rev{Q_{\vec{0}}}\circ\C{(3,0)}\circ Q_{\vec{1}})~\circ\\
		&\hphantom{=~~}      \C{(6d+9,0),(6d+12,0)}\text{.}
	\end{align*}
	\noindent
	In both cases, the curves~$P_\u$ and~$Q_\v$ have distance~$1$ if~$\u$ and~$\v$ are orthogonal and distance~$3$ otherwise.
	For a \YES-instance with orthogonal vectors~$\u_i$ and~$\v_j$, match the first copy of~$P_\mathit{skip}$ with the first~$1+2j$ gadgets of type~$Q_\v$.
	Similarly, match the last copy of~$P_\mathit{skip}$ to the last~$2(m-j)-1$ gadgets of type~$Q_\v$.
	Match the copy of~$\rev{Q_{\vec{0}}}$ preceding~$Q_{\v_j}$ with~$P$ up until the gadget~$P_{\u_i}$, match~$P_{\u_i}$ with~$Q_{\v_j}$ and match the copy of~$\rev{Q_{\vec{0}}}$ after~$Q_{\v_j}$ starting after the gadget~$P_{\u_i}$ of~$P$. 
	This yields a matching of width~$1$.
	See Figure~\ref{fig:weak} (Left).

	Conversely, a matching of width less than~$3$ must traverse one of the curves~$P_{\u_i}$ and~$Q_{\v_j}$ simultaneously, which is not possible for a \NO-instance.
	In the construction, any matching of width less than~$3$ can be extended into one containing~$(1,1)$ and~$(|P|,|Q|)$.
	Hence, the reductions also apply to the weak Fr\'echet distance without endpoint restrictions.
	\begin{figure}\centering%
		\includegraphics[width=\textwidth]{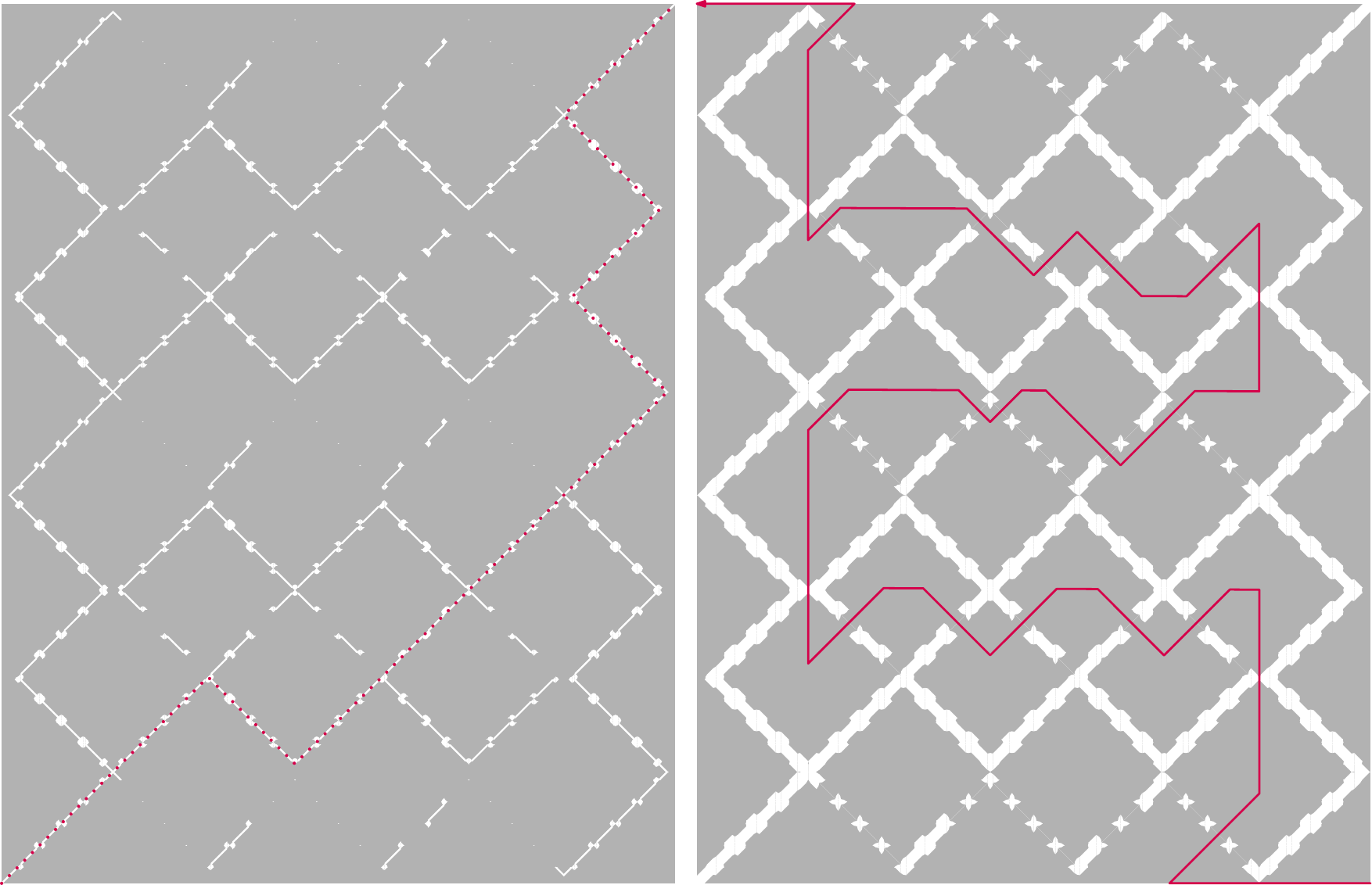}%
		\caption{Left: the~$1$-freespace of a \YES-instance with a matching (dotted). Right: the~$3$-freespace of a \NO-instance with a cut.\label{fig:weak}}%
	\end{figure}
	\begin{theorem}
		For any polynomial restriction of~$1\leq |P|\leq|Q|$, the discrete weak Fr\'echet distance between one-dimensional curves~$P$ and~$Q$ (with and without endpoint restrictions) and the weak Fr\'echet distance (with and without endpoint restrictions) between two dimensional curves~$P$ and~$Q$, has no~$\bigO((|P||Q|)^{1-\delta})$ time~$(3-\e)$-approximation unless \SETH' fails.\label{thm:partialDiscrete}
	\end{theorem}

\subsection{Continuous one-dimensional weak Fr\'echet distance}\label{sec:algoWeak}
	In this section, we show that the continuous weak Fr\'echet distance can be computed in linear time for curves in one dimension.
	For this case, the weak Fr\'echet distance without endpoint restrictions was already known to be computable in linear time, namely because it is equivalent to the Hausdorff distance between the images of those curves~\cite{buchin2007difficult}.
	
	Our algorithm for computing the continuous weak Fr\'echet distance is more complicated.
	It will be helpful to transform the input curves into canonical ones.
	Let a \emph{canonical} curve be a continuous one-dimensional curve~$P$ that contains no four consecutive vertices~$\{p_a,p_{a+1},p_{a+2},p_{a+3}\}$ with~$p_a\leq p_{a+2}\leq p_{a+1}\leq p_{a+3}$ or~$p_{a+3}\leq p_{a+1}\leq p_{a+2}\leq p_a$.
	By repeating the transformation of Lemma~\ref{lem:weakSimplify}, one can in linear time transform any continuous one-dimensional curve~$P$ into a canonical curve~$P'$ with~$\dwF(P,P')=0$.
	\begin{lemma}\label{lem:weakSimplify}
		Let~$P$ be a continuous one-dimensional curve that is not canonical due to vertices~$\{p_a,p_{a+1},p_{a+2},p_{a+3}\}$.
		Let~$P'$ be the copy of~$P$ with the edges between~$p_a$ and~$p_{a+3}$ replaced by a single edge, then~$\dwF(P,P')=0$.
	\end{lemma}
	\begin{proof}
		Assume that~$p_{a+1}\neq p_{a+2}$ (otherwise we are done).
		Then~$p'_i=p_i$ for~$i\leq a$ and~$p'_i=p_{i+2}$ for~$i\geq a+1$.
		Pick~$a_1$ and~$a_2$ such that~$a\leq a_2<a_1\leq a+1$ and~$P'(a_2)=p_{a+2}$ and~$P'(a_1)=p_{a+1}$.
		Then the piecewise linear path with vertex sequence~$(1,1)$,~$(a,a)$,~$(a+1,a_1)$,~$(a+2,a_2)$,~$(a+3,a+1)$, and~$(|P|,|P'|)$ is a weak Fr\'echet matching of width~$0$.
	\end{proof}

	\noindent
	Canonical curves have the following structural properties.
	\begin{lemma}
		For a canonical curve~$P$ and any~$i<i'$, any shortest edge of~$P[i,i']$ has~$p_i$ or~$p_{i'}$ as an endpoint.
	\end{lemma}
	\begin{proof}
		Otherwise the shortest edge of~$P[i,i']$ is surrounded by edges that are at least as long and hence form a witness that~$P$ is not canonical.
	\end{proof}

	\begin{corollary}
		For any canonical curve, the subsequence of local maxima is quasiconcave, the subsequence of local minima is quasiconvex, and the vertices that are global minima or maxima (of which there are at most three) are all consecutive.
	\end{corollary}

	\noindent
	A \emph{growing} curve is any canonical curve whose last edge contains both a global maximum and minimum.
	For a point~$p$ and a curve~$Q$, let~$d(p,Q)$ be the distance from~$p$ to the closest point on~$Q$.
	Our algorithm for the weak Fr\'echet distance uses the following linear-time subroutine.
	\newcommand{\GreedyMatching}{\textbf{GreedyMatching}}
	\begin{algorithm}\label{alg:greedymatching}%
	\begin{minipage}[t]{0pt}\bfseries%
		\begin{tabbing}~~\=~~~\=~~\=~~~\=~~\=~~~\=~~\=~~~\=\kill
			\GreedyMatching$(P,Q)$:\\
			$|$\>\> $r$ \>= $|p_1-q_1|$\\
			$|$\>\> $i$ \>= $1$\\
			$|$\>\> $j$ \>= $1$\\
			$|$\>\> while $i+1<|P|$ or $j+1<|Q|$:\\
			$|$\>\>$|$\>\> if $j+1<|Q|$ and $d(q_{j+1},P[i,i+1])\leq r$:\\
			$|$\>\>$|$\>\>$|$\>\> $j$ \>= $j+1$\\
			$|$\>\>$|$\>\> else if $i+1<|P|$:\\
			$|$\>\>$|$\>\>$|$\>\> $r$ \>= $\max(r,d(p_{i+1},Q[j,j+1]))$\\
			$|$\>\>$|$\>\>$|$\>\> $i$ \>= $i+1$\\
			$|$\>\>$|$\>\> else if $j+1<|Q|$:\\
			$|$\>\>$|$\>\>$|$\>\> $r$ \>= $\max(r,d(q_{j+1},P[i,i+1]))$\\
			$|$\>\>$|$\>\>$|$\>\> $j$ \>= $j+1$\\
			$|$\>\> return $r$
		\end{tabbing}%
	\end{minipage}%
	\end{algorithm}
	\begin{lemma}
		Algorithm~\ref{alg:greedymatching} computes for growing curves~$P$ and~$Q$, the minimum width over all paths~$\Phi\subseteq[1,|P|]\times[1,|Q|]$ from~$(1,1)$ to any~$(x,y)\in[|P|-1,|P|]\times[|Q|-1,|Q|]$.
	\end{lemma}
	\begin{proof}
		Let~$M(i,j)=[i,i+1]\times[j,j+1]\cup(i+1,|P|]\times[1,|Q|]\cup[1,|P|]\times(j+1,|Q|]$.
		We use as invariant that~(1) there is a path~$\Phi\subseteq[1,|P|]\times[1,|Q|]$ of width~$r$ from~$(1,1)$ to some~$(x,y)\in[i,i+1]\times[j,j+1]$, and that~(2) there are no paths of width less than~$r$ to any~$(x,y)\in M(i,j)$.
		Indeed, this invariant holds at the start of the loop.
		When the algorithm returns we have~$i+1=|P|$ and~$j+1=|Q|$ as desired, so it remains to show that the invariant is maintained.
		Part~(1) of the invariant is maintained by construction, so it suffices to show that part~(2) is maintained.
		
		Fix some~$(i,j,r)$ and suppose that the invariant is satisfied.
		It will clearly be maintained for the next iteration if~$d(q_{j+1},P[i,i+1])\leq r$ or~$d(p_{i+1},Q[j,j+1])\leq r$, as any path to~$M(i+1,j)$ or~$M(i,j+1)$ must also enter~$M(i,j)$.
		So assume that~$r<d(q_{j+1},P[i,i+1])$ and~$r<d(p_{i+1},Q[j,j+1])$.
		Suppose for a contradiction that the invariant does not hold for the values~$(i',j',r')$ of~$(i,j,r)$ after the next iteration, then there is a path~$\Phi$ of width~$s$ with~$r\leq s<r'\leq\min(d(q_{j+1},P[i,i+1]),d(p_{i+1},Q[j,j+1]))$ from~$(1,1)$ to a point outside~$[1,i+1]\times[1,j+1]$.
		Let~$(x',y')$ be the point where~$\Phi$ leaves~$[1,i+1]\times[1,j+1]$.
		Then either~$x'=i+1$ or~$y'=j+1$.
		If~$x'=i+1$, then~$s\geq |P(x')-Q(y')|\geq d(p_{i+1},Q[1,j+1])\geq d(p_{i+1},Q[j,j+1])\geq r'$.
		Similarly, if~$y'=j+1$, then~$s\geq |P(x')-Q(y')|\geq d(q_{j+1},P[1,i+1])\geq d(q_{j+1},P[i,i+1])\geq r'$.
		As both cases give a contradiction, the invariant is maintained.
	\end{proof}

	\begin{lemma}\label{lem:weakCanonical}
		The weak Fr\'echet distance for canonical curves is computable in linear time.	
	\end{lemma}
	\begin{proof}
		Consider canonical curves~$P$ and~$Q$.
		If one curve has a single vertex, the weak Fr\'echet distance is its distance to the furthest point on the other curve.
		Consider an edge~$P[i,i+1]$ between a global minimum and maximum of~$P$.
		Define the growing curves~$P_L=P[1,i+1]$, and~$P_R=\rev{P[i,|P|]}$.
		Similarly, consider such an edge~$Q[j,j+1]$ of~$Q$ and define~$Q_L$ and~$Q_R$ analogously.
		Let~$r_L=\GreedyMatching(P_L,Q_L)$ and~$r_R=\GreedyMatching(P_R,Q_R)\}$.
		
		We show that~$\dwF(P,Q)=\max(r_L,r_R)$.
		For this, consider a weak Fr\'echet matching~$\Phi$ of width~$w$ between~$P$ and~$Q$.
		Define~$\pi_{i+1}^P\from[1,|P|]\to[i,i+1]$ as the map for which~$\pi_{i+1}(x)$ is the unique point~$x'\in[i,i+1]$ with~$P(x')=P(x)$ for~$x>i+1$, and~$\pi_{i+1}(x)=x$ for~$x\leq i+1$.
		Define a path~$\Phi'$ by replacing any point~$(x,y)\in\Phi$ by~$(\pi_{i+1}^P(x),\pi_{j+1}^Q(y))$.
		Then~$\Phi'\subseteq[1,i+1]\times[1,j+1]$ is a path of width at most~$w$ from~$(1,1)$ to a point~$(x,y)\in[i,i+1]\times[j,j+1]$.
		Thus~$r_L\leq\dwF(P,Q)$ and by symmetric argument~$r_R\leq\dwF(P,Q)$.

		It remains to show that~$\dwF(P,Q)\leq\max(r_L,r_R)$.
		There exists a path~$\Phi_L$ of width~$r_L$ from~$(1,1)$ to~$(x_L,y_L)\in[i,i+1]\times[j,j+1]$.
		There also is a path~$\Phi_R$ of width~$r_R$ from~$(x_R,y_R)\in[i,i+1]\times[j,j+1]$ to~$(|P|,|Q|)$.
		Connecting these paths with the straight segment from~$(x_L,y_L)$ to~$(x_R,y_R)$ yields a weak Fr\'echet matching of the desired width.
	\end{proof}

	\begin{theorem}\label{thm:weak}
		The weak Fr\'echet distance between continuous one-dimensional curves can be computed in linear time.
	\end{theorem}
	\begin{proof}
		Transform input curves~$P$ and~$Q$ into canonical curves~$P'$ and~$Q'$ in linear time.
		By triangle inequality we have~$\dwF(P,Q)=\dwF(P',Q')$, which can be computed in linear time by Lemma~\ref{lem:weakCanonical}.
	\end{proof}

\section{Discussion}
	We have shown that the Fr\'echet and many of its variants cannot be approximated better than factor~$3$ in strongly subquadratic time unless~\SETH{}' fails.
	Although we show that similar reductions cannot improve upon this factor, it remains open whether this factor is tight, or if there is a strongly subquadratic constant factor approximation at all.
	Furthermore, for curves in~$1$D, our construction for the Fr\'echet distance does not rule out a strongly subquadratic algorithm for curves with an imbalanced number of vertices.

\bibliography{refs}

\newcommand{\SortNoop}[1]{}
\begin{thebibliography}{10}

\bibitem{AbboudBDN18}
Amir Abboud, Karl Bringmann, Holger Dell, and Jesper Nederlof.
\newblock More consequences of falsifying {SETH} and the orthogonal vectors
  conjecture.
\newblock In {\em Proc. 50th Sympos. Theory Comput. (STOC)}, pages 253--266,
  2018.

\bibitem{agarwal2014computing}
Pankaj~K Agarwal, Rinat~Ben Avraham, Haim Kaplan, and Micha Sharir.
\newblock Computing the discrete {F}r{\'e}chet distance in subquadratic time.
\newblock {\em SIAM J. Comput.}, 43(2):429--449, 2014.

\bibitem{AltBuchin10}
Helmut Alt and Maike Buchin.
\newblock Can we compute the similarity between surfaces?
\newblock {\em Discrete Comput. Geom. (DCG)}, 43(1):78--99, 2010.

\bibitem{altgodau}
Helmut Alt and Michael Godau.
\newblock Computing the {F}r\'echet distance between two polygonal curves.
\newblock {\em Internat. J. Comput. Geom. Appl. (IJCGA)}, 5(1--2):75--91, 1995.

\bibitem{AronovHKWW06}
Boris Aronov, Sariel Har-Peled, Christian Knauer, Yusu Wang, and Carola Wenk.
\newblock Fr{\'e}chet distances for curves, revisited.
\newblock In {\em Proc. 14th European Sympos. Algorithms (ESA)}, pages 52--63,
  2006.

\bibitem{brakatsoulas2005map}
Sotiris Brakatsoulas, Dieter Pfoser, Randall Salas, and Carola Wenk.
\newblock On map-matching vehicle tracking data.
\newblock In {\em Proc. 31st International Conference on Very Large Data
  Bases}, pages 853--864. VLDB Endowment, 2005.

\bibitem{Bringmann14}
Karl Bringmann.
\newblock Why walking the dog takes time: {F}r\'echet distance has no strongly
  subquadratic algorithms unless {SETH} fails.
\newblock In {\em Proc. 55th Sympos. Found. Comput. Sci. (FOCS)}, pages
  661--670, 2014.

\bibitem{Bringmann17}
Karl Bringmann and Marvin K{\"{u}}nnemann.
\newblock Improved approximation for {F}r{\'{e}}chet distance on $c$-packed
  curves matching conditional lower bounds.
\newblock {\em Internat. J. Comput. Geom. Appl. (IJCGA)}, 27(1-2):85--120,
  2017.

\bibitem{BringmannM16}
Karl Bringmann and Wolfgang Mulzer.
\newblock Approximability of the discrete {F}r{\'{e}}chet distance.
\newblock {\em J. Computational Geometry (JoCG)}, 7(2):46--76, 2016.

\bibitem{BuchinBDFJSSSW17}
Kevin Buchin, Maike Buchin, David Duran, Brittany~Terese Fasy, Roel Jacobs,
  Vera Sacrist{\'{a}}n, Rodrigo~I. Silveira, Frank Staals, and Carola Wenk.
\newblock Clustering trajectories for map construction.
\newblock In {\em Proc. 25thInternat. Conf. on Advances in Geographic
  Information Systems (SIGSPATIAL)}, pages 14:1--14:10, 2017.

\bibitem{buchin2007difficult}
Kevin Buchin, Maike Buchin, Christian Knauer, G{\"u}nter Rote, and Carola Wenk.
\newblock How difficult is it to walk the dog.
\newblock In {\em Proc. 23rd European Workshop on Computational Geometry
  (EuroCG)}, pages 170--173, 2007.

\bibitem{buchin2017soviets}
Kevin Buchin, Maike Buchin, Wouter Meulemans, and Wolfgang Mulzer.
\newblock Four soviets walk the dog: Improved bounds for computing the
  {F}r{\'e}chet distance.
\newblock {\em Discrete Comput. Geom. (DCG)}, 58(1):180--216, 2017.

\bibitem{BuchinBuSc10}
Kevin Buchin, Maike Buchin, and Andr{\'e} Schulz.
\newblock {F}r{\'e}chet distance of surfaces: Some simple hard cases.
\newblock In {\em Proc. 18th European Sympos. Algorithms (ESA)}, pages 63--74,
  2010.

\bibitem{BuchinBW09}
Kevin Buchin, Maike Buchin, and Yusu Wang.
\newblock Exact algorithms for partial curve matching via the {F}r{\'e}chet
  distance.
\newblock In {\em Proc. 20th Sympos. Discrete Algorithms (SODA)}, pages
  645--654, 2009.

\bibitem{BuchinCLMMOS17}
Kevin Buchin, Jinhee Chun, Maarten L{\"{o}}ffler, Aleksandar Markovic, Wouter
  Meulemans, Yoshio Okamoto, and Taichi Shiitada.
\newblock Folding free-space diagrams: Computing the {F}r{\'{e}}chet distance
  between 1-dimensional curves (multimedia contribution).
\newblock In {\em Proc. 33rd Sympos. Comput. Geom. (SoCG)}, pages 64:1--64:5,
  2017.

\bibitem{buchin2015computing}
Kevin Buchin, Tim Ophelders, and Bettina Speckmann.
\newblock Computing the similarity between moving curves.
\newblock In {\em Proc. 23rd European Sympos. Algorithms (ESA)}, pages
  928--940, 2015.

\bibitem{CookWenk10}
Atlas~F. Cook and Carola Wenk.
\newblock Geodesic {F}r{\'e}chet distance inside a simple polygon.
\newblock {\em Transactions on Algorithms}, 7(1):Art. 9, 2010.

\bibitem{DriemelH13}
Anne Driemel and Sariel Har{-}Peled.
\newblock Jaywalking your dog: Computing the fr{\'{e}}chet distance with
  shortcuts.
\newblock {\em SIAM J. Comput.}, 42(5):1830--1866, 2013.

\bibitem{DriemelHW12}
Anne Driemel, Sariel Har{-}Peled, and Carola Wenk.
\newblock Approximating the {F}r{\'{e}}chet distance for realistic curves in
  near linear time.
\newblock {\em Discrete {\&} Computational Geometry}, 48(1):94--127, 2012.

\bibitem{DriemelKS16}
Anne Driemel, Amer Krivosija, and Christian Sohler.
\newblock Clustering time series under the {F}r{\'{e}}chet distance.
\newblock In {\em Proc. 27th Sympos. Discrete Algorithms (SODA)}, pages
  766--785, 2016.

\bibitem{eiter1994computing}
Thomas Eiter and Heikki Mannila.
\newblock Computing discrete fr{\'e}chet distance.
\newblock Technical Report CD-TR 94/64, Information Systems Department,
  Technical University of Vienna, 1994.

\bibitem{GudmundssonWolle10}
Joachim Gudmundsson and Thomas Wolle.
\newblock {Towards automated football analysis: Algorithms and data
  structures}.
\newblock In {\em Proc. 10th Australasian Conf. on Mathematics and Computers in
  Sport}, 2010.

\bibitem{har2014frechet}
Sariel Har-Peled and Benjamin Raichel.
\newblock The {F}r{\'e}chet distance revisited and extended.
\newblock {\em Transactions on Algorithms}, 10(1):3, 2014.

\bibitem{kane2017orthogonal}
Daniel Kane and Ryan Williams.
\newblock The orthogonal vectors conjecture for branching programs and
  formulas.
\newblock {\em arXiv preprint arXiv:1709.05294}, 2017.

\bibitem{konzack2017visual}
Maximilian Konzack, Thomas McKetterick, Tim Ophelders, Maike Buchin, Luca
  Giuggioli, Jed Long, Trisalyn Nelson, Michel~A Westenberg, and Kevin Buchin.
\newblock Visual analytics of delays and interaction in movement data.
\newblock {\em Internat. J. of Geographical Information Science},
  31(2):320--345, 2017.

\bibitem{NayyeriX16}
Amir Nayyeri and Hanzhong Xu.
\newblock On computing the {F}r{\'{e}}chet distance between surfaces.
\newblock In {\em Proc. 32nd Sympos. Comput. Geom. (SoCG)}, pages 55:1--55:15,
  2016.

\bibitem{Williams04}
Ryan Williams.
\newblock A new algorithm for optimal constraint satisfaction and its
  implications.
\newblock In {\em Proc. 31st Internat. Colloq. Automata Lang. Program.
  (ICALP)}, pages 1227--1237, 2004.

\end{thebibliography}
\newpage
\appendix
\section{Figures accompanying Section~\ref{sec:frechet}}
	\begin{figure*}[h]\centering%
		\makebox[\linewidth]{\includegraphics{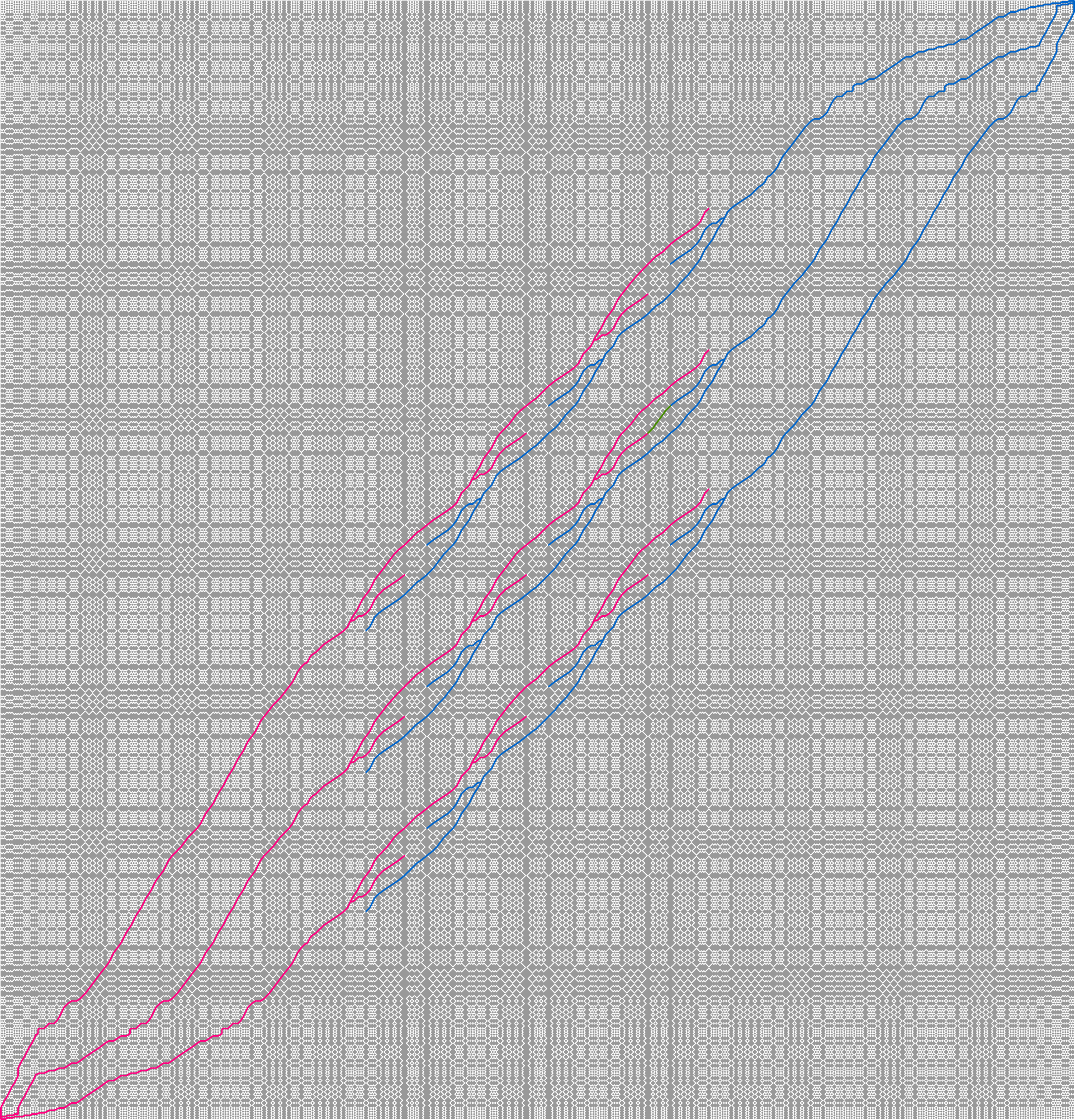}}%
		\caption{The~$1$-free space containing a matching of width~$1$ for our construction for a \YES-instance with~$U=\{(1,1,0),(0,1,1),(1,0,1)\}$ and~$V=\{(0,1,0),(1,0,1),(1,1,0)\}$.\label{fig:fdSat}}
		\vspace{-15em}
	\end{figure*}
	\begin{figure*}[h]\centering%
		\makebox[\linewidth]{\includegraphics{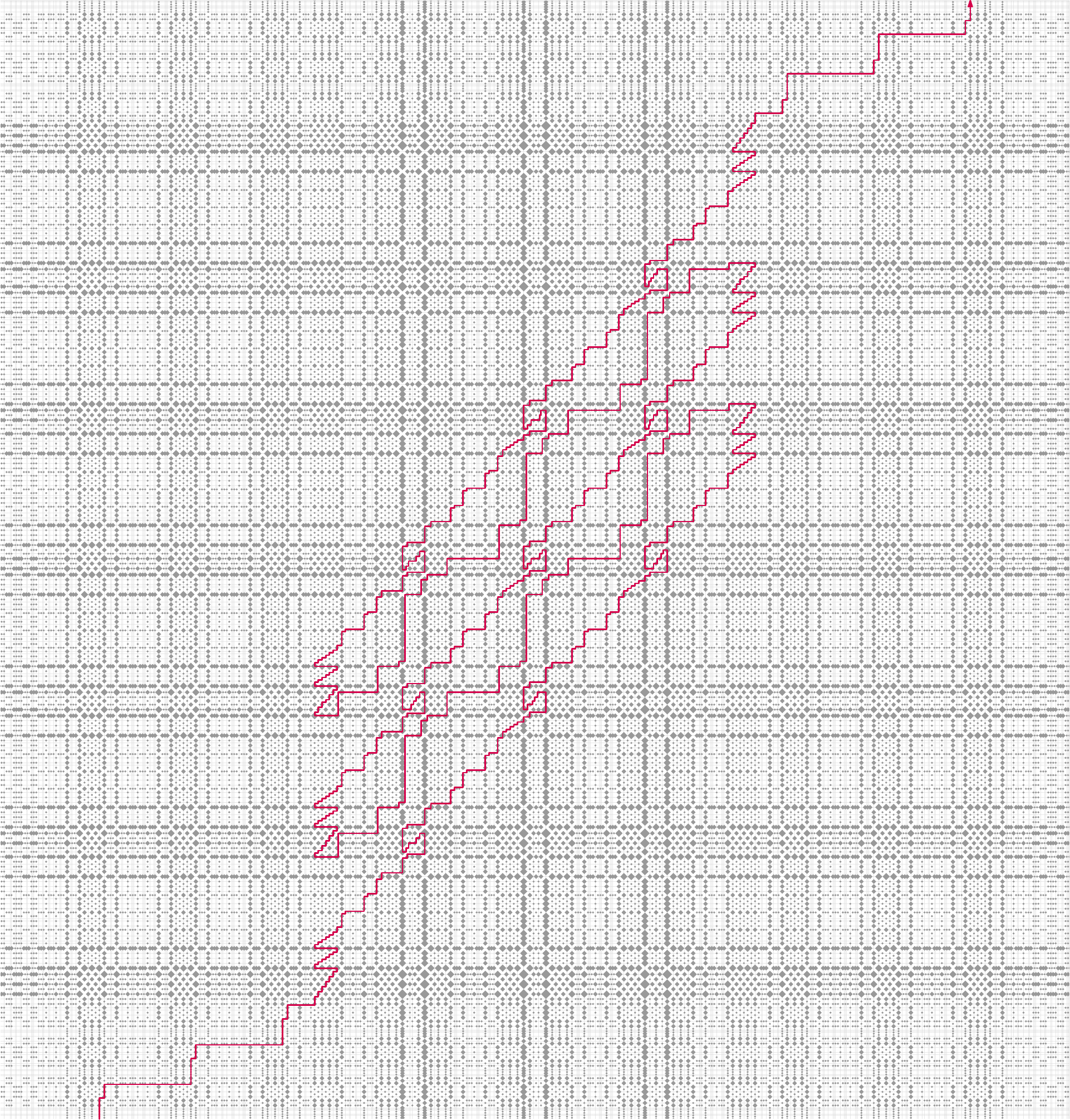}}%
		\caption{A cut of width~$3$ drawn in the~$3$-free space for our construction for a \NO-instance with~$U=\{(1,1,0),(0,1,1),(1,0,1)\}$ and~$V=\{(0,1,1),(1,0,1),(1,1,0)\}$.\label{fig:fdUnsat}}
	\end{figure*}
\end{document}